\pdfoutput=1
\documentclass[
	bigMargin = false,
	final,
	font = palatino
	]{MyClass}

\usepackage{myMathHeader}
\usepackage{autonum}

\crefname{subsection}{Subsection}{Subsections}

\let\ConnSpace\relax
\let\GauGroup\relax
\let\GauAlgebra\relax
\let\SymTensorFieldSpace\relax
\let\MetricSpace\relax
\let\DiffGroup\relax
\let\VolSpace\relax
\MyNewMathOperator{\ConnSpace}		{command={\mathbrush{C}}, sort={C}, display={$\ConnSpace(P)$}, description={Space of connections of the principal bundle $P$}}
\MyNewMathOperator{\GauGroup}			{command={\mathbrush{G}}, sort={G}, display={$\GauGroup(P)$}, description={Group of gauge transformation of the (principal) bundle $P$}}
\MyNewMathOperator{\GauAlgebra}		{command={\mathrm{L}\mathbrush{G}}, sort={LG}, display={$\GauAlgebra(P)$}, description={Lie algebra of infinitesimal gauge transformation of the principal bundle $P$}}
\MyNewMathOperator{\SymTensorFieldSpace}		{command={\mathrm{S}}}
\MyNewMathOperator{\MetricSpace}		{command={\mathbrush{M}}}
\MyNewMathOperator{\DiffGroup}		{command={\mathbrush{D}}}
\MyNewMathOperator{\VolSpace}		{command={\mathbrush{V}}}

\RenewDocumentCommand{\wedgeDual}{s m m}{#2 \mathbin{\dot{\wedge}} #3}

\Title{Clebsch-Lagrange variational principle and geometric constraint analysis of relativistic field theories}

\Abstract{
	Inspired by the Clebsch optimal control problem, we introduce a new variational principle that is suitable for capturing the geometry of relativistic field theories with constraints related to a gauge symmetry.
	Its special feature is that the Lagrange multipliers couple to the configuration variables via the symmetry group action.
	The resulting constraints are formulated as a condition on the momentum map of the gauge group action on the phase space of the system.
	We discuss the Hamiltonian picture and the reduction of the gauge symmetry by stages in this geometric setting.
	We show that the Yang--Mills--Higgs action and the Einstein--Hilbert action fit into this new framework after a \( (1+3) \)-splitting.
	Moreover, we recover the Gauß constraint of Yang--Mills--Higgs theory and the diffeomorphism constraint of general relativity as momentum map constraints.
}



\Institution{itp}{
	Institut für Theoretische Physik, 
	Universität Leipzig, 
	04009 Leipzig, Germany	}
\Institution{mpiitp}{
	Max Planck Institute for Mathematics in the Sciences,
	04103 Leipzig, Germany
	and
	Institut für Theoretische Physik, 
	Universität Leipzig, 
	04009 Leipzig, Germany	}

\Author[
	affiliation = mpiitp,
	email = tobias.diez@mis.mpg.de,
	]{diez}{Tobias Diez}
\Author[
	affiliation = itp,
	email = {rudolph@rz.uni-leipzig.de}
	]{rudolph}{Gerd Rudolph}

\begin{document}

\MakeTitle

\tableofcontents

\listoftodos

\section{Introduction}

Relativistic field theories of gauge symmetry type can be formulated as constrained Hamiltonian systems.
Using a \( (1 + 3) \)-decomposition of spacetime, the four dimensional field equations usually split into hyperbolic evolution equations and elliptic constraint equations for the Cauchy data.
A remarkable observation is that the constraints often can be formulated in terms of the momentum map associated to the action of the symmetry group of the theory.
This is usually verified by a fairly routine calculation in each example separately (see, \eg, \parencite{ArmsMarsdenEtAl1981,Arms1981} for pure Yang--Mills theory and \parencite{FischerMarsden1979,ArmsFischerMarsden1975} for general relativity).
Going beyond a case-by-case study, the general philosophy that the subset of the phase space cut out by the constraints can be identified with the zero set of a momentum map seems to be true in a large number of models.
In this paper, we will argue that the relationship between constraints and momentum maps is not a lucky coincidence but is rooted in a special feature of the four dimensional physical action.
This special feature is the main subject of the present paper.

Inspired by the Clebsch optimal control problem \parencite{Gay-BalmazRatiu2011}, we study a variational principle associated to a class of degenerate Lagrangians, whose degeneracy results from the action of the symmetry group.
The defining feature of this principle is that the Lagrange multipliers are elements of the Lie algebra of the symmetry group and that they couple to the configuration variables via a given Lie algebra action.
The resulting equations of motion, which we call the Clebsch--Euler--Lagrange equations, decompose into an evolution and a constraint equation.
Next, we define a Clebsch--Legendre transformation similar to the ordinary Legendre transformation leading to what we call the Clebsch--Hamiltonian picture.
As we will show, the constraints that arise from the degeneracy of the Lagrangian are phrased as momentum map constraints on the Hamiltonian side.

We then use the geometric formulation of the Dirac--Bergmann algorithm in the formulation of \textcite{GotayNesterHinds1978} to derive the Clebsch--Legendre transformation from an ordinary Legendre transformation of an extended Lagrangian system.
In this process, we isolate two constraint equations, which are later given a geometric interpretation in terms of a symplectic reduction by stages procedure.
This analysis of the constraints extends the well-known results for Yang--Mills theory \parencite{BergveltDeKerf1986,BergveltDeKerf1986a} and general relativity \parencite{Giulini2015}.

Since constraints which allow a reformulation in terms of symmetric Hamiltonian systems are abundant in field theories, this observation suggests that the Clebsch--Lagrange principle is fundamental for Cauchy problems.
The two examples that will be discussed are the Yang--Mills--Higgs equations and the Einstein equation.
In both cases, we will complete the following program:
\begin{enumerate}
	\item 
		Using a \( (1+3)\)-splitting of spacetime, formulate the field equations as a Cauchy problem and determine the constraints on the initial data.
		In particular, find a symplectic manifold on which the dynamics takes place.
	\item
		Identify the Lagrange multipliers as elements of the Lie algebra of the symmetry group and determine their action on the phase space.
	\item
		Calculate the Lagrangian in the \( (1+3) \)-splitting and show that it is of the Clebsch--Lagrange form.
	\item
		Pass to the Clebsch--Hamiltonian picture using the Clebsch--Legendre transformation.
		In particular, phrase the constraints in terms of the momentum map.
	\item
		Relate the Clebsch--Hamiltonian picture to the Hamiltonian system with constraints obtained by the ordinary Legendre transformation. 
	\item
		Discuss the symmetry reduction as a symplectic reduction by stages.
\end{enumerate}
After accomplishing this program, we are left with a singular symplectic cotangent bundle reduction in infinite dimensions, which is studied in detail in a separate paper \parencite{DiezRudolphReduction}.

\paragraph*{Acknowledgments}
We thank the anonymous reviewer for his/her comments and suggestions on an earlier draft, which significantly contributed to improving the quality of the paper.
We gratefully acknowledge support of the Max Planck Institute for Mathematics in the Sciences in Leipzig and of the University of Leipzig.

%%%%%%%%%%%%%%%%%%%%%%%%%%%%%%%%%%%%%%%%%%%%%%%%%%%%%%%%%%%%%%%%%%%%%%%%%%%%%%%%%%%
%%%%%%%%%%%%%%%%%%%%%%%%%%%%%%%%%%%%%%%%%%%%%%%%%%%%%%%%%%%%%%%%%%%%%%%%%%%%%%%%%

\section{Clebsch--Lagrange variational principle}
\label{sec:clebschLagrange}

%%%%%%%%%%%%%%%%%%%%%%%%%%%%%%%%%%%%%%%%%%%%%%%%%%%%%%%%%%%%%%%%%%%%%%%%%%%%%%%%%
%%%%%%%%%%%%%%%%%%%%%%%%%%%%%%%%%%%%%%%%%%%%%%%%%%%%%%%%%%%%%%%%%%%%%%%%%%%%%%%%%

In \parencite{Gay-BalmazRatiu2011}, \citeauthor{Gay-BalmazRatiu2011} study a class of optimal control problems associated to group actions.
The special feature of the \emphDef{Clebsch optimal control problem} is that the control variables are Lie algebra valued and couple to the state variables via the symmetry group action. 
Let \( G \) be a Lie group that acts smoothly on the smooth manifold \( Q \) and let \( \LieA{g} \) be its Lie algebra.
Given a smooth cost function \( l: Q \times \LieA{g} \to \R \), the Clebsch optimal control problem consists in finding curves \( t \mapsto q(t) \) and \( t \mapsto \xi(t) \) such that 
\begin{equation}
	s[q, \xi] = \int_0^T l\bigl(q(t), \xi(t)\bigr) \dif t
\end{equation}
is minimized subject to the optimal control constraint \( \dot q(t) = - \xi(t) \ldot q(t) \) and the endpoint constraints \( q(0) = q_i \) and \( q(T) = q_f \).
Here, \( \xi \ldot q \) is the fundamental vector field at \( q \in Q \) generated by the action of \( \xi \in \LieA{g} \).
The Pontryagin maximum principle shows that the control variable \( \xi \) satisfies the (generalized) Euler-Poincaré equation, see \parencite[Theorem~7.1]{Gay-BalmazRatiu2011}.
This observation allows one to give an optimal control formulation for many systems including the heavy top and the compressible or magnetohydrodynamic fluid flow.

%%%%%%%%%%%%%%%%%%%%%%%%%%%%%%%%%%%%%%%%%%%%%%%%%%%%%%%%%%%%%%%%%%%%%%%%%

\subsection{Clebsch--Lagrangian picture}

%%%%%%%%%%%%%%%%%%%%%%%%%%%%%%%%%%%%%%%%%%%%%%%%%%%%%%%%%%%%%%%%%%%%%%%%%%

We now describe a novel variational principle which draws its inspiration from the Clebsch optimal control problem and which will turn out to be useful for relativistic field theories with constraints.
The idea is to dismiss the optimal control constraint and to treat, instead, the quantity \( \dot q(t) + \xi(t) \ldot q(t) \) as an effective velocity in the Lagrangian formulation.
As above, let \( Q \) be a smooth manifold and let \( \TBundle Q \) be its tangent bundle.
Assume that a Lie group \( G \) acts smoothly on \( Q \).
The configuration space \( Q \) as well as the Lie group \( G \) are assumed to be (infinite-dimensional) Fréchet manifolds.
We refer to \parencite{Neeb2006,Hamilton1982} for the differential calculus in Fréchet spaces.
In order to stay close to physics notation, we will usually denote points in \( \TBundle Q \) by pairs \( (q, \dot{q}) \), where \( q \in Q \) stands for the base point of the vector \( \dot{q} \in \TBundle_q Q \).
Sometimes, we will also write a point in \( \TBundle Q \) as a pair \( (q, v) \) with \( q \in Q \) and \( v \in \TBundle_q Q \).
Moreover, we use the dot notation \( g \cdot q \) for the action of \( g \in G \) on \( q \in Q \).
For the derivative of the action a similar notation using lower dots is employed, \ie, the fundamental vector field \( \xi_* \) generated by \( \xi \in \LieA{g} \) is written as \( \xi_* (q) = \xi \ldot q \in \TBundle_q Q \) and the lifted \( G \)-action on \( \TBundle Q \) has the form \( g \ldot X_q \in \TBundle_{g \cdot q} Q \) for \( X_q \in \TBundle_q Q \).

Given a smooth function \( L: \TBundle Q \times \LieA{g} \to \R \), we consider the variational principle \( \diF S = 0 \) for curves \( t \mapsto q(t) \in Q \) and \( t \mapsto \xi(t) \in \LieA{g} \), where the physical action \( S \) is of the form
\begin{equation}\label{eq:clebschLagrange:action}
	S[q, \xi] = \int_0^T L\bigl(q(t), \dot q(t) + \xi(t) \ldot q(t), \xi(t) \bigr) \dif t
\end{equation}
and where the variations of \(t \mapsto q(t) \) are restricted to vanish at the endpoints.
Here, \(t \mapsto (q (t), \dot q (t) + \xi(t) \ldot q (t)) \) is viewed as a curve in \( \TBundle Q \).
We will refer to this variational principle as the \emphDef{Clebsch--Lagrange principle} and to \( L \) as the \emphDef{Clebsch--Lagrangian}.

In order to derive the associated equations of motions, we first recall some notions and results from geometric mechanics (see, \eg \parencite{AbrahamMarsdenEtAl1980, RudolphSchmidt2012}).
In particular, we need the theory of linear connections in vector bundles.
A linear connection in a vector bundle \( \pi: E \to M \) is a vector bundle homomorphism \( K: \TBundle E \to E \)  over \( \pi: E \to M \) and over \( \TBundle M \to M \) such that the following digram commutes:
\begin{equationcd}
	\VBundle E \to[r, "K"] 
		& E \to[d, "\id_E"] 
	\\
	E \times_M E \to[r, "\pr_2"] \to[u, "\mathrm{vl}"]
		& E,
\end{equationcd}
where the vertical tangent bundle \( \VBundle E = \ker \tangent \pi \) is identified with \( E \times_M E \) via the linear structure on \( E \), that is,
\begin{equation}
	\mathrm{vl}: E \times_M E \to \VBundle E, \quad (e, v) \mapsto \difFracAt{}{\varepsilon}{0} (e + \varepsilon v) \in \VBundle_e E.
\end{equation}
The kernel \( \HBundle E \equiv \ker K \) of \( K \) is then a vector subbundle of \( \TBundle E \) such that \( \TBundle E = \HBundle E \oplus \VBundle E \) is a fiberwise topological isomorphism.
Accordingly, the connection \( K \) yields a bundle isomorphism
\begin{equation}
	\label{eq:connection:decomposition}
	\TBundle E \to E \times_M \TBundle M \times_M E, \quad (Z_e) \mapsto \bigl(e, \tangent_e \pi(Z_e), K(Z_e)\bigr).
\end{equation}
The inverse is given by
\begin{equation}
	\label{eq:connection:decompositionInverse}
	E \times_M \TBundle M \times_M E \to \TBundle E, \quad (e, X, v) \mapsto \mathrm{vl}(e, v) + X^K_e,
\end{equation}
where \( X^K_e \) denotes the \( K \)-horizontal lift of \( X \) to \( \TBundle_e E \).
Given a smooth section \( \phi \) of \( E \), its covariant derivative \( \nabla \phi \) is defined by
\begin{equation}
	\label{eq:connection:covDeriv}
	\nabla_X \phi = K \bigl(\tangent \phi (X)\bigr)
\end{equation}
for every \( X \in \VectorFieldSpace(M) \).
Using the Leibniz rule, \( \nabla \) is extended to the exterior covariant derivative \( \dif_K: \DiffFormSpace^k(M, E) \to \DiffFormSpace^{k+1}(M, E) \) of \( E \)-valued differential forms.
A soldering form on \( E \) is a vector bundle morphism \( \vartheta: \TBundle M \to E \) (usually, \( \vartheta \) is required to be an isomorphism of vector bundles, but we do not need this assumption).
The tangent bundle \( E = \TBundle M \) comes equipped with a tautological soldering form given by the identity \( \vartheta: \TBundle M \to \TBundle M \).
Given a soldering form \( \vartheta \in \DiffFormSpace^1(M, E) \), the torsion of a connection \( K \) on \( E \) relative to \( \vartheta \) is defined to be \( T = \dif_K \vartheta \in \DiffFormSpace^2(M, E) \).
In particular, we have
\begin{equation}
	T (X, Y)
		= \dif_K \vartheta (X, Y)
		= \nabla_X \bigl(\vartheta (Y)\bigr) - \nabla_Y \bigl(\vartheta (X)\bigr) - \vartheta (\commutator{X}{Y}),
\end{equation}
which, for the tangent bundle \( E = \TBundle M \) with the tautological soldering form, reduces to the usual defining relation for the torsion.
Let \( f: N \to M \) be a smooth map, and let \( E \to M \) be a vector bundle endowed with a connection \( K \).
The pull-back bundle \( f^* E \) is a smooth vector bundle over \( N \).
The connection on \( E \) yields via pull-back a connection in \( f^* E \) and thus induces a covariant derivative
\begin{equation}
 	\nabla^f_Y: \sSectionSpace(f^* E) \to \sSectionSpace(f^* E),
 	\quad \psi \mapsto K\bigl(\tangent \psi (Y)\bigr)
\end{equation} 
for every \( Y \in \VectorFieldSpace(N) \).
In particular, for a curve \( \gamma: \R \supseteq I \to M \), we write
\begin{equation}
	\label{eq:connection:alongCurve}
 	\DifFrac{}{t} \equiv \nabla^\gamma_{\difp_t}: \sSectionSpace(\gamma^* E) \to \sSectionSpace(\gamma^* E)
\end{equation}
for the covariant derivative along \( \gamma \) in the direction of the canonical vector field \( \difp_t \) on \( I \).

Returning to our original setting, let \( Q \) be a smooth manifold and let \( L: \TBundle Q \to \R \) be a smooth Lagrangian.
The \emphDef{fiber or vertical derivative} \( \difFibre L: \TBundle Q \to \TBundle' Q \) of \( L \) is defined by 
\begin{equation}
	\label{FibDer}
	\dualPair{\difFibre L (v)}{w} = \difFracAt{}{\varepsilon}{0} L(v + \varepsilon w)\, ,
\end{equation}
where \( v, w \in \TBundle_q Q \), and \( \TBundle' Q \) denotes the fiberwise dual of the tangent bundle.
In order to stay close to the usual physics notation, the fiber derivative \( \difFibre L: \TBundle Q \to \TBundle' Q \) will be written as \( \difpFrac{L}{\dot q} \).
We stress that the latter symbol will always stand for the mapping given by~\eqref{FibDer},
\begin{equation}
	\dualPair*{\difpFrac{L}{\dot{q}}(q, v)}{w}
		= \difFracAt{}{\varepsilon}{0} L(q, v + \varepsilon w) \,, 
		\qquad w \in \TBundle_q Q.
\end{equation}
In order to define the partial derivative \( \difpFrac{L}{q} \) in an intrinsic way, we need a linear connection in \( \TBundle Q \).
By~\eqref{eq:connection:decomposition}, every connection \( K: \TBundle (\TBundle Q) \to \TBundle Q \) in the bundle \( \TBundle Q \to Q \) yields a vector bundle isomorphism
\begin{equation}
	\label{eq:clebschLagrange:decompositionTTwo}
	\TBundle (\TBundle Q) \isomorph \TBundle Q \times_Q \TBundle Q \times_Q \TBundle Q.
\end{equation}
Accordingly, we will write elements of \( \TBundle (\TBundle Q) \) as tuples \( (q, \dot{q}, \diF q, \diF \dot{q}) \), where \( \dot{q}, \diF q \) and \( \diF \dot{q} \) are elements of \( \TBundle_q Q \).
We emphasize that the isomorphism~\eqref{eq:clebschLagrange:decompositionTTwo} and, in particular, the component \( \diF \dot{q} \) depend on the connection \( K \).
Now, for \( L: \TBundle Q \to \R \), the derivative of \( L \) at \( (q, \dot{q}) \in \TBundle Q \) with respect to the second component in the decomposition~\eqref{eq:clebschLagrange:decompositionTTwo} will be written as \( \difpFrac{L}{q}(q, \dot{q}): \TBundle Q \to \R \).
In contrast to the fiber derivative \( \difpFrac{L}{\dot{q}} \), the partial derivative \( \difpFrac{L}{q} \) depends on the choice of a connection in \( \TBundle Q \).
In summary, we express \( \tangent L: \TBundle (\TBundle Q) \to \R \) as
\begin{equation}
	\label{eq:clebschLagrange:tangentLagrange}
	\tangent L(q, \dot{q}, \diF q, \diF \dot{q}) = \dualPair*{\difpFrac{L}{q}(q, \dot{q})}{\diF q} + \dualPair*{\difpFrac{L}{\dot{q}}(q, \dot{q})}{\diF \dot{q}}.
\end{equation}
This relation is the intrinsic geometric version of the corresponding coordinate expression found in the physics literature.
For \( L: \TBundle Q \times \LieA{g} \to \R \), the derivative of \( L \) at \( (q, \dot q, \xi) \in \TBundle Q \times \LieA{g} \) in the \( \LieA{g} \)-direction yields an element of the topological dual \( \LieA{g}' \) and will be denoted by \( \difpFrac{L}{\xi}(q, \dot q, \xi) \).

In infinite dimensions, one needs to be careful with the notion of a cotangent bundle, because the canonical candidate --- the fiberwise topological dual \( \TBundle' Q \) of the tangent bundle --- fails to be a \emph{smooth} bundle.
Following \parencite{DiezRudolphReduction}, we define the cotangent bundle \( \CotBundle Q \) to be some smooth vector bundle\footnote{
For example, we may take \( \CotBundle Q \) to be \( \TBundle Q \) and the pairing given by a Riemannian structure on \( Q \).
In applications, the fiber of \( \TBundle Q \) is often a space of mappings so that a convenient choice of the cotangent bundle consists of regular distributions inside the space of all distributions.
} over \( Q \) which is fiberwise in duality with the tangent bundle relative to some chosen pairing \( \dualPairDot: \CotBundle Q \times \TBundle Q \to \R \).
In line with our notation for points of \( \TBundle Q \), we will denote points in the cotangent bundle \( \CotBundle Q \) by pairs  \( (q, p) \) with \( q \in Q \) and \( p \in \CotBundle_{q} Q \).
The dual pairing yields an embedding of the cotangent bundle \( \CotBundle Q \) into the topological dual \( \TBundle' Q \) of the tangent bundle.
In the following, we always assume that all occurring dual objects can be represented by points of the cotangent bundle \( \CotBundle Q \).
In particular, the derivatives \( \difpFrac{L}{q} \) and \( \difpFrac{L}{\dot{q}} \) are viewed as maps \( \TBundle Q \to \CotBundle Q \) (and not merely to \( \TBundle' Q \)).
Regularity of \( L \) then means that \( \difpFrac{L}{\dot{q}} \) is a diffeomorphism between \( \TBundle Q \) and \( \CotBundle Q \).
Similarly as for the cotangent bundle, we choose a non-degenerate pairing \( \kappa: \LieA{g}^* \times \LieA{g} \to \R \).
This gives an embedding of \( \LieA{g}^* \) into the topological dual \( \LieA{g}' \).
We silently assume that all occurring dual objects that a priori are only in \( \LieA{g}' \) may actually be represented by elements of \( \LieA{g}^* \).
This remark applies in particular to the partial derivative \( \difpFrac{L}{\xi} \).

In order to formulate the variational principle, we also need the following technical tool.
A \emphDef{local addition} on \( Q \) is a smooth map \( \eta: \TBundle Q \supseteq U \to Q \) defined on an open neighborhood \( U \) of the zero section in \( \TBundle Q \) such that 	the composition of \( \eta \) with the zero section is the identity on \( Q \), \ie, \( \eta(q, 0) = q \), and 	the map \( \pr_Q \times \eta: \TBundle Q \supseteq U \to Q \times Q \) is a diffeomorphism onto an open neighborhood of the diagonal, \cf \parencite[Section~42.4]{KrieglMichor1997}.
	A local addition may be constructed using the exponential map of a Riemannian metric on \( Q \) but it also exists if \( Q \) is an affine space, a Lie group or a space of sections of a fiber bundle.
	
With these preliminaries out of the way, we are now able to state the equations of motion corresponding to the Clebsch--Lagrange variational principle.
\begin{thm}\label{prop:clebschLagrange:clebschEulerLagrangeDirect}
	Let \( Q \) be a smooth Fréchet \( G \)-manifold.
	Assume that \( Q \) can be endowed with a local addition\footnote{More generally, one could give up the assumption of the existence of a local addition and instead work in the diffeological category \parencite{Iglesias-Zemmour2013}. 
	}.
	Moreover, assume that \( \TBundle Q \) and \( \CotBundle Q \) are endowed with dual torsion-free linear connections.
	Then, for a Clebsch--Lagrangian \( L: \TBundle Q \times \LieA{g} \to \R \) the following are equivalent:
	\begin{thmenumerate}
		\item
			The curves \( t \mapsto q(t) \in Q \) and \( t \mapsto \xi(t) \in \LieA{g} \) are solutions of the variational Clebsch--Lagrange problem~\eqref{eq:clebschLagrange:action}.
		\item
			The following equations hold:
			\begin{subequations}\label{eq:clebschLagrange:clebschEulerLagrangeDirect}
			\begin{align+}
				\DifFrac{}{t} \left( \difpFrac{L}{\dot{q}} \right) - \dualPair*{\difpFrac{L}{\dot{q}}}{\nabla \xi_*} - \difpFrac{L}{q} &= 0,
				\label{eq:clebschLagrange:clebschEulerLagrangeDirect:evol}
				\\
				\dualPair*{\difpFrac{L}{\dot{q}}}{\zeta \ldot q} + \kappa\left(\difpFrac{L}{\xi}, \zeta\right) &= 0,
				\label{eq:clebschLagrange:clebschEulerLagrangeDirect:constraint} 
			\end{align+}
			\end{subequations}
			for all \( \zeta \in \LieA{g} \), where evaluation of all derivatives at \( (q, \dot q + \xi \ldot q, \xi) \in \TBundle Q \times \LieA{g} \) is understood.
			Furthermore, the second term in~\eqref{eq:clebschLagrange:clebschEulerLagrangeDirect:evol} denotes the functional \( T_q Q \ni X \mapsto \dualPair[\big]{\difpFrac{L}{\dot{q}}}{\nabla_X \xi_*} \in \R \) viewed as an element of \( \CotBundle Q \).
			\qedhere
	\end{thmenumerate}
\end{thm}
The equations~\eqref{eq:clebschLagrange:clebschEulerLagrangeDirect} will be referred to as the \emphDef{Clebsch--Euler--Lagrange equations}.

\begin{proof}
	Let \( I = [0, T] \) with \( T > 0 \).
	Since \( Q \) admits a local addition, the space \( \sFunctionSpace_{q_i q_f}(I, Q) \) of paths in \( Q \) with given endpoints \( q_i \) and \( q_f \) is a smooth infinite-dimensional manifold, \cf \parencite[Theorem~7.6]{Wockel2014}.
	Similarly, as \( \LieA{g} \) is an affine space, the space \( \sFunctionSpace(I, \LieA{g}) \) of paths in \( \LieA{g} \) is a smooth manifold, too.
	Since the Lagrangian is smooth, the action \( S \) is a smooth function on \( \sFunctionSpace_{q_i q_f}(I, Q) \times \sFunctionSpace(I, \LieA{g}) \).
	The variational principle amounts to looking for extrema of \( S \).

	Let \( \bigl(q_\varepsilon(t), \dot{q}_\varepsilon(t)\bigr) \) be a smooth perturbation of the curve \( \bigl(q(t), \dot{q}(t)\bigr) \) in \( \TBundle Q \).
	In line with usual notation, we write
	\begin{equation}
		\label{eq:clebschLagrange:clebschEulerLagrangeDirect:variation}
		 \difFracAt{}{\varepsilon}{0}\bigl(q_\varepsilon(t), \dot{q}_\varepsilon(t)\bigr) \equiv \bigl(\diF q(t), \diF \dot{q}(t)\bigr) \in \TBundle_{q(t)} Q \times \TBundle_{q(t)} Q
	\end{equation}
	for the corresponding curve in \( \TBundle(\TBundle Q) \) under the isomorphism~\eqref{eq:clebschLagrange:decompositionTTwo}.
	The curve \( t \mapsto \bigl(\diF q(t), \diF \dot{q}(t)\bigr) \) should be considered as a tangent vector field along the curve \( t \mapsto q(t) \), which may be viewed as an tangent vector to \( \sFunctionSpace_{q_i q_f}(I, Q) \) at that curve.
	Let \( \phi: I \times \R \to Q \) be defined by \( \phi(t, \varepsilon) = q_\varepsilon(t) \), and let \( \vartheta \) be the tautological soldering form on \( \TBundle Q \).
	Note that \( (\phi^* \vartheta)_{t, \varepsilon} (\difp_t) = \dot{q}_\varepsilon (t) \) and \( (\phi^* \vartheta)_{t, \varepsilon} (\difp_\varepsilon) = \diF q_\varepsilon (t) \).
	Since the connection \( K \) on \( \TBundle Q \) is torsion-free with respect to \( \vartheta \), using~\eqref{eq:connection:alongCurve} and~\eqref{eq:clebschLagrange:clebschEulerLagrangeDirect:variation}, we obtain
	\begin{equation}\begin{split}
		0 &= \phi^* T (\difp_t, \difp_\varepsilon)
			\\
			&= \dif_{\phi^* K} (\phi^* \vartheta) (\difp_t, \difp_\varepsilon)
			\\
			&= \nabla^\phi_{\difp_t} \bigl(\phi^* \vartheta (\difp_\varepsilon)\bigr) - \nabla^\phi_{\difp_\varepsilon} \bigl(\phi^* \vartheta (\difp_t)\bigr) - \phi^* \vartheta \bigl(\commutator{\difp_t}{\difp_\varepsilon}\bigr)
			\\
			&= \DifFrac{}{t} (\diF q) - \diF \dot{q}(t).
	\end{split}\end{equation}
	That is, \( \diF \dot{q}(t) = \DifFrac{}{t} (\diF q) \).
	Moreover, by~\eqref{eq:connection:covDeriv}, we have 
	\begin{equation}
		K \left(\difFracAt{}{\varepsilon}{0} \xi \ldot q_\varepsilon\right)
			= K \bigl(\tangent \xi_* (\diF q)\bigr)
			= \nabla_{\diF q} \xi_*,
	\end{equation}
	where \( \xi_* \) is viewed as a section \( Q \to \TBundle Q \) and where we suppressed the explicit time dependence.
	In passing, we note that \( \tangent \xi_* (\diF q) \) is the fundamental vector field at \( \diF q \) on \( \TBundle Q \) generated by \( \xi \).
	Thus, under the isomorphism~\eqref{eq:clebschLagrange:decompositionTTwo}, we obtain
	\begin{equation}\begin{split}
		\difFracAt{}{\varepsilon}{0}\bigl(q_\varepsilon(t), \dot{q}_\varepsilon(t) + \xi(t) \ldot q_\varepsilon(t)\bigr)
			&= \Bigl(\diF q(t), \diF \dot{q}(t) + \nabla_{\diF q (t)} \xi(t)_* \Bigr)
			\\
			&= \left(\diF q(t), \DifFrac{}{t} (\diF q) + \nabla_{\diF q (t)} \xi(t)_* \right).
	\end{split}\end{equation}
	Thus, using~\eqref{eq:clebschLagrange:tangentLagrange}, variation of the path in \( Q \) yields
	\begin{align}
		&\int_0^T \left( \dualPair*{\difpFrac{L}{q}}{\diF q} +  \dualPair*{\difpFrac{L}{\dot{q}}}{\DifFrac{}{t} (\diF q)} + \dualPair*{\difpFrac{L}{\dot{q}}}{\nabla_{\diF q} \xi_*} \right) \dif t = 0,
		\label{eq:clebschLagrange:proof:geom}
		\intertext{which after integration by parts on the second term is equivalent to}
		& \int_0^T \left(\dualPair*{\difpFrac{L}{q}}{\diF q} - \dualPair*{\DifFrac{}{t} \left(\difpFrac{L}{\dot{q}}\right)}{\diF q} + \dualPair*{\difpFrac{L}{\dot{q}}}{\nabla_{\diF q} \xi_*} \right)  \dif t = 0.
	\end{align}
	Since the variations \( \diF q \) are arbitrary, we get the evolution equation~\eqref{eq:clebschLagrange:clebschEulerLagrangeDirect:evol}.

	Similarly, let \( t \mapsto \diF \xi(t) \in \LieA{g} \) be a tangent vector to \( t \mapsto \xi(t) \) in \( \sFunctionSpace(I, \LieA{g}) \).
	Variation of \( \xi \) yields
	\begin{align}
		\dualPair*{\difpFrac{L}{\dot{q}}}{\diF \xi \ldot q} + \kappa\left(\difpFrac{L}{\xi}, \diF \xi\right) = 0,
	\end{align}
	which clearly gives~\eqref{eq:clebschLagrange:clebschEulerLagrangeDirect:constraint}.
\end{proof}

\begin{remark}[Clebsch optimal control problem]
	\label{rem:clebschLagrange:optimalControl}
	Suppose the Clebsch--Lagrange problem is supplemented by the optimal control constraint \( \dot q = - \xi \ldot q \).
	Then,~\eqref{eq:clebschLagrange:action} reduces to the variational principle \( \diF s = 0 \) for the action
	\begin{equation}
		s[q, \xi] = \int_0^T l(q, \xi) \dif t,
	\end{equation}
	where the effective Lagrangian is given by \( l(q, \xi) = L(q, 0, \xi) \).
	We hence recover the Clebsch optimal control problem discussed at the beginning of the section.

	For the applications we have in mind, the optimal control constraint is a too strong condition.
	In the context of Yang--Mills theory, it amounts to requiring that the color-electromagnetic field vanishes and, for Einstein's equation, it is equivalent to a zero extrinsic curvature.
\end{remark}

%%%%%%%%%%%%%%%%%%%%%%%%%%%%%%%%%%%%%%%%%%%%%%%%%%%%%%%%%%%%%%%%%%%%%%%%%%%%%%%%%

\subsection{Clebsch--Hamilton picture}
\label{MomMapCon}

%%%%%%%%%%%%%%%%%%%%%%%%%%%%%%%%%%%%%%%%%%%%%%%%%%%%%%%%%%%%%%%%%%%%%%%%%%%%%%%%%%

In this section, we discuss the Hamiltonian counterpart of the Clebsch--Euler--Lagrange equations.
Let \( E: \TBundle Q \times \LieA{g} \to \R \),
\begin{equation}
	E(q, \dot{q}, \xi) \defeq \dualPair*{\difpFrac{L}{\dot{q}}(q, \dot{q} + \xi \ldot q, \xi)}{\dot{q} + \xi \ldot q} - L(q, \dot{q} + \xi \ldot q, \xi),
\end{equation}
be the energy function associated to the Clebsch--Lagrangian \( L \).
The \emphDef{Clebsch--Legendre transformation} is defined by
\begin{equation}
	\label{eq:clebschLagrange:clebschLegendre}
	\mathrm{CL}: \TBundle Q \times \LieA{g} \to \CotBundle Q \times \LieA{g},
		\quad
		(q, \dot{q}, \xi) \mapsto \left(q, \difpFrac{L}{\dot{q}}(q, \dot{q} + \xi \ldot q, \xi), \xi \right).
\end{equation}
We say that \( L \) is regular if \( \mathrm{CL} \) is a diffeomorphism.
For a regular Clebsch--Lagrangian \( L \), the associated \emphDef{Clebsch--Hamiltonian} \( H: \CotBundle Q \times \LieA{g} \to \R \) is defined by \( H = E \circ \mathrm{CL}^{-1} \).
\begin{remark}
	Clearly \( \mathrm{CL} \) is a diffeomorphism if and only if the fiber derivative
	\begin{equation}
	 	\TBundle Q \to \CotBundle Q,
	 	\quad 
	 	(q, v) \mapsto \left(q, \difpFrac{L}{\dot{q}}(q, v, \xi)\right)
	 \end{equation}
	is a diffeomorphism for every \( \xi \in \LieA{g} \).
	Moreover, the Clebsch--Hamiltonian \( H \) coincides with the Hamiltonian corresponding to the Clebsch--Lagrangian \( L \) via the ordinary Legendre transformation with \( \xi \) viewed as a parameter, that is, 
	\begin{equation}\label{eq:hamiltonian}
		H \left(q, p , \xi \right) = \dualPair*{p}{v} - L(q, v, \xi)\, ,
	\end{equation}
	where the relation \( p = \difpFrac{L}{\dot{q}}(q, v, \xi) \) defines \( v \) as a function of \( q, p \) and \( \xi \).
\end{remark}
Let \( K \) be a torsion-free connection in \( \TBundle Q \) and let \( \bar{K} \) be the dual connection in \( \CotBundle Q \).
By~\eqref{eq:connection:decomposition}, \( \bar{K} \) induces a decomposition
\begin{equation}
	\label{eq:clebschLagrange:hamilton:decompCotBundle}
	\TBundle (\CotBundle Q) \isomorph \CotBundle Q \times_Q \TBundle Q \times_Q \CotBundle Q
\end{equation}
and we will write points in \( \TBundle (\CotBundle Q) \) as tuples \( (q, p, \delta q, \delta p) \) with \( p, \delta p \in \CotBundle_q Q \) and \( \delta q \in \TBundle_q Q \).
Note that the component \( \delta p \) depends on the choice of the connection \( \bar{K} \).
Given a smooth function \( H: \CotBundle Q \times \LieA{g} \to \R \) (not necessarily obtained via a Clebsch--Legendre transformation), we decompose the derivative \( \tangent H: \TBundle(\CotBundle Q) \to \R \) relative to~\eqref{eq:clebschLagrange:hamilton:decompCotBundle} as follows:
\begin{equation}
	\label{eq:clebschLagrange:hamilton:decompTHamilton}
	\tangent H (q, p, \delta q, \delta p) 
		= \dualPair*{\difpFrac{H}{q}(q, p)}{\delta q} + \dualPair*{\delta p}{\difpFrac{H}{p}(q, p)},
\end{equation}
where we assume that the partial derivatives are maps \( \difpFrac{H}{q}: \CotBundle Q \to \CotBundle Q \) and \( \difpFrac{H}{p}: \CotBundle Q \to \TBundle Q \).
We emphasize that, according to~\eqref{eq:connection:decompositionInverse}, \( \difpFrac{H}{p} \) is the intrinsically defined fiber derivative of \( H \) while \( \difpFrac{H}{q} \) depends on the choice of the connection \( \bar{K} \).
\begin{prop}[Equations under Clebsch--Legendre transformation]
	\label{prop:clebschLagrange:legendreDirect}
	For a regular Clebsch--Lagrangian \( L: \TBundle Q \times \LieA{g} \to \R \) with associated Clebsch--Hamiltonian \( H: \CotBundle Q \times \LieA{g} \to \R \),
	%Then, under the Legendre transformation, the Clebsch--Lagrange equations are equivalent to the Clebsch--Hamilton equations.
	the curves \( t \mapsto q(t) \in Q \) and \( t \mapsto \xi(t) \in \LieA{g} \) are solutions of the Clebsch--Lagrange variational problem~\eqref{eq:clebschLagrange:action} if and only if the curves
	\begin{equation}
		t \mapsto (q(t), p(t), \xi(t)) = \mathrm{CL}\bigl(q(t), \dot{q}(t), \xi(t)\bigr)
	\end{equation}
	satisfy the following equations
	\begin{subequations}\label{eq:clebschLagrange:hamiltonianDirect}\begin{gather+}
		\label{eq:clebschLagrange:hamiltonianDirect:hamiltonian}
		\difpFrac{H}{q}(q, p, \xi) = - \DifFrac{}{t} p - \bar{K}(\xi \ldot p), \qquad \difpFrac{H}{p}(q, p, \xi) = \dot q + \xi \ldot q,
		\\
		\label{eq:clebschLagrange:hamiltonianDirect:constraint}
		\dualPair*{p}{\zeta \ldot q} = \kappa\left(\difpFrac{H}{\xi}(q, p, \xi), \zeta\right),
	\end{gather+}\end{subequations}
	for all \( \zeta \in \LieA{g} \).
\end{prop}
\begin{proof}
	The proof is by direct inspection.
	Let the curve \( t \mapsto p(t) \) be defined by \( p(t) \defeq \difpFrac{L}{\dot{q}} \bigl(q, \dot q + \xi \ldot q, \xi \bigr)(t) \).
	Using the Clebsch--Euler--Lagrange equation~\eqref{eq:clebschLagrange:clebschEulerLagrangeDirect:evol}, we have
	\begin{equation}\label{eq:clebschLagrange:hamiltonian:proof:pTime}
		\DifFrac{}{t} p = \dualPair{p}{\nabla \xi_*} + \difpFrac{L}{q}(q, \dot q + \xi \ldot q, \xi). 	
	\end{equation}
	On the other hand, using the identification~\eqref{eq:hamiltonian} of \( H \) as the Legendre transform of \( L \) with \( v = \dot q + \xi \ldot q \), we get
	\begin{equation}
		\difpFrac{H}{q}(q, p, \xi)
			= \dualPair*{p}{\difpFrac{v}{q}} - \dualPair*{\difpFrac{L}{\dot q}}{\difpFrac{v}{q}} - \difpFrac{L}{q}(q, v, \xi)
			= - \difpFrac{L}{q}(q, \dot q + \xi \ldot q, \xi).
	\end{equation}
	Comparing with~\eqref{eq:clebschLagrange:hamiltonian:proof:pTime}, we see that~\eqref{eq:clebschLagrange:clebschEulerLagrangeDirect:evol} is equivalent to
	\begin{equation}
		\difpFrac{H}{q}(q, p, \xi) = - \DifFrac{}{t} p + \dualPair{p}{\nabla \xi_*}.
	\end{equation}
	Let us rewrite the last term on the right-hand side of this equation.
	For \( \xi \in \LieA{g} \), choose a curve \( \varepsilon \mapsto g_\varepsilon \) in \( G \) with \( g_0 = e \) and \( \difFracAt{}{\varepsilon}{0} g_\varepsilon = \xi \).
	Let \( \Upsilon: \R \to Q \), \( \Upsilon(\varepsilon) = g_\varepsilon \cdot q \) be the orbit curve through \( q \in Q \).
	Considering \( \varepsilon \mapsto g_\varepsilon \ldot \dot{q} \) as a section of \( \Upsilon^* \TBundle Q \) for \( \dot{q} \in \TBundle_q Q \), we have
	\begin{equation}
		\nabla_{\difp_\varepsilon}^\Upsilon (g_\varepsilon \ldot \dot{q})|_{\varepsilon = 0}
			= K \left(\difFracAt{}{\varepsilon}{0} g_\varepsilon \ldot \dot{q}\right)
			= K (\xi \ldot \dot{q})
			= K \bigl(\tangent \xi_* (\dot{q})\bigr)
			= \nabla_{\dot{q}} \xi_* \, .
	\end{equation}
	A similar calculation yields \( \bar{\nabla}_{\difp_\varepsilon}^\Upsilon (g_\varepsilon \ldot p)|_{\varepsilon = 0} = \bar{K}(\xi \ldot p) \) for \( p \in \CotBundle_q Q \).
	Thus, taking the derivative of the identity \( \dualPair{p}{\dot{q}} = \dualPair{g_\varepsilon \cdot p}{g_\varepsilon \ldot \dot{q}} \) with respect to \( \varepsilon \) gives
	\begin{equation}\label{eq:clebschLagrange:hamiltonian:proof:nablaXi}\begin{split}
		0 
		&= \difFracAt{}{\varepsilon}{0} \dualPair{g_\varepsilon \cdot p}{g_\varepsilon \ldot \dot{q}}
		\\
		&= \dualPair*{p}{\nabla_{\difp_\varepsilon}^\Upsilon (g_\varepsilon \ldot \dot{q})|_{\varepsilon = 0}} + \dualPair*{\bar{\nabla}_{\difp_\varepsilon}^\Upsilon (g_\varepsilon \ldot p)|_{\varepsilon = 0}}{\dot{q}}
		\\
		&= \dualPair{p}{\nabla_{\dot{q}} \xi_*} + \dualPair{\bar{K}(\xi \ldot p)}{\dot{q}}.
	\end{split}\end{equation}
	That is \( \dualPair{p}{\nabla \xi_*} = - \bar{K}(\xi \ldot p) \).
	Hence, in summary, the first equation in~\eqref{eq:clebschLagrange:hamiltonianDirect:hamiltonian} is equivalent to~\eqref{eq:clebschLagrange:clebschEulerLagrangeDirect:evol}.

	Moreover,~\eqref{eq:hamiltonian} implies
	\begin{equation}
		\difpFrac{H}{p}(q, p, \xi)
			= v + \dualPair*{p}{\difpFrac{v}{p}} - \dualPair*{\difpFrac{L}{\dot q}(q, v, \xi)}{\difpFrac{v}{p}}
			= \dot q + \xi \ldot q,
	\end{equation}
	which yields the second equation in~\eqref{eq:clebschLagrange:hamiltonianDirect:hamiltonian}.

	Similarly, the derivative of \( H \) in the \( \xi \)-direction is given by
	\begin{equation}\begin{split}
		\difpFrac{H}{\xi}(q, p, \xi)
			&= \dualPair*{p}{\difpFrac{v}{\xi}} - \dualPair*{\difpFrac{L}{\dot q}(q, v, \xi)}{\difpFrac{v}{\xi}} - \difpFrac{L}{\xi}(q, v, \xi)
			\\ 
			&= - \difpFrac{L}{\xi}(q, \dot q + \xi \ldot q, \xi).
	\end{split}\end{equation}
	Hence, the constraint~\eqref{eq:clebschLagrange:hamiltonianDirect:constraint} is equivalent to~\eqref{eq:clebschLagrange:clebschEulerLagrangeDirect:constraint}.
\end{proof}
\begin{remark}
	In the proof of \cref{prop:clebschLagrange:legendreDirect}, we have seen that the first equation in~\eqref{eq:clebschLagrange:hamiltonianDirect:hamiltonian} is equivalent to
	\begin{equation+}
		\difpFrac{H}{q}(q, p, \xi) = - \DifFrac{}{t} p + \dualPair{p}{\nabla \xi_*}.
		\qedhere
	\end{equation+}
\end{remark}
The constraint equation~\eqref{eq:clebschLagrange:hamiltonianDirect:constraint} has a natural reformulation in terms of the momentum map of the lifted \( G \)-action on the cotangent bundle \( \CotBundle Q \).
As usual, \( \CotBundle Q \) carries a canonical \( 1 \)-form \( \theta \) defined by 
\begin{equation}
	\label{CanForm}
	\theta_{(q,p)} (X) = \dualPair{p}{\tangent_{(q,p)} \pi (X)}_q,
\end{equation}
where \( X \in \TBundle_{(q,p)} (\CotBundle Q) \) and \( \pi \) is the natural projection of \( \CotBundle Q \).
Correspondingly, the canonical symplectic form is \( \omega = \dif \theta \).
Recall the \( G \)-action on \( \TBundle Q \) written as \( g \ldot \dot{q} \in \TBundle_{g \cdot q} Q \) for \( g \in G \) and \( \dot{q} \in \TBundle_q Q \).
We assume that the relation
\begin{equation}
	\dualPair{g \cdot p}{\dot{q}}
		= \dualPair{p}{g^{-1} \ldot \dot{q}}
\end{equation}
for \( g \in G, \dot{q} \in \TBundle_q Q \) and \( p \in \CotBundle_q Q \), defines a lift to \( \CotBundle Q \) of the \( G \)-action on \( Q \) (which is automatic in finite dimensions).
The lifted action on \( \CotBundle Q \) is symplectic, and the associated \( G \)-equivariant momentum map 
\( J: \CotBundle Q \to \LieA{g}^* \) (if it exists) is defined by
\begin{equation}
	\label{eq:cotangentBundle:momentumMapDef}
	\kappa\left(J(q, p), \xi\right)
		= \dualPair{p}{\xi \ldot q}.
\end{equation}
In infinite dimensions, the momentum map may not exist in pathological cases, \cf \parencite[Example~2.2]{DiezRudolphReduction}, and we thus have to assume its existence in the forthcoming. 
The right-hand side of~\eqref{eq:cotangentBundle:momentumMapDef} is exactly what occurs in the constraint~\eqref{eq:clebschLagrange:hamiltonianDirect:constraint} and we hence arrive at the following reformulation of the Clebsch--Euler--Lagrange equations~\eqref{eq:clebschLagrange:clebschEulerLagrangeDirect}.
\begin{coro}\label{prop:clebschLagrange:legendre}
	In the setting of~\cref{prop:clebschLagrange:legendreDirect}, assume additionally that the lifted \( G \)-action on the cotangent bundle \( \CotBundle Q \) has a momentum map \( J: \CotBundle Q \to \LieA{g}^* \) (which is automatic in finite dimensions).
	Then, the Clebsch--Euler--Lagrange equations are equivalent to the following set of equations:
	\begin{subequations}\label{eq:clebschLagrange:hamiltonian}\begin{gather+}
		\label{eq:clebschLagrange:hamiltonian:hamiltonian}
		\difpFrac{H}{q}(q, p, \xi) = - \DifFrac{}{t} p - \bar{K}(\xi \ldot p), \qquad \difpFrac{H}{p}(q, p, \xi) = \dot q + \xi \ldot q,
		\\
		\label{eq:clebschLagrange:hamiltonian:constraint}
		J(q, p) = \difpFrac{H}{\xi}(q, p, \xi).
	\end{gather+}\end{subequations}
\end{coro}
\noindent
The system~\eqref{eq:clebschLagrange:hamiltonian} on \( \CotBundle Q \times \LieA{g} \) will be referred to as the \emphDef{Clebsch--Hamilton equations}.
We will refer to~\eqref{eq:clebschLagrange:hamiltonian:constraint} as the \emphDef{momentum map constraint}.

It turns out that the Clebsch--Hamilton equations can be written in a form where the dynamics is given by a Hamiltonian vector field (relative to the canonical symplectic structure on \( \CotBundle Q \)).
As this needs some additional technical insight, we will discuss this aspect in a separate paper.

\begin{example}
	A special model of a Clebsch--Hamiltonian system is studied in \parencite{GayBalmazHolmRatiu2013}.
	Let \( (Q, h) \) be a Riemannian manifold, which we assume to be finite-dimensional for simplicity.
	Moreover, let \( G \) be a (finite-dimensional) Lie group acting on \( Q \).
	Consider the Lagrangian
	\begin{equation}
		L(q, \dot q, \xi) = \frac{m}{2} \norm{\dot q}^2 - V(q, \xi),
	\end{equation}
	where the norm is taken with respect to the Riemannian metric \( h \).
	Thus, the Clebsch--Lagrange problem~\eqref{eq:clebschLagrange:action} consists in minimizing the action functional
	\begin{equation}
		S[q, \xi] = \int_0^T \left(\frac{m}{2} \norm{\dot q + \xi \ldot q}^2 - V(q, \xi)\right).
	\end{equation}
	This variational problem is investigated in \parencite[Section~3.2]{GayBalmazHolmRatiu2013}.
	By~\eqref{eq:hamiltonian}, the associated Clebsch--Hamiltonian is given by
	\begin{equation}
		H(q, p, \xi) = \frac{1}{2m} \norm{\dot q}^2 + V(q, \xi).
	\end{equation}
	Thus, the Clebsch--Hamilton equations~\eqref{eq:clebschLagrange:hamiltonian} take the following form (with respect to the Levi--Civita connection):
	\begin{subequations}
		\begin{align}
			\DifFrac{}{t} p + \bar{K}(\xi \ldot p) &= - \difpFrac{V}{q},
			\\
			p^\sharp &= m (\dot q + \xi \ldot q),
			\\
			J(q, p) &= \difpFrac{V}{\xi},
		\end{align}
	\end{subequations}
	where \( p^\sharp \in \TBundle_q Q \) denotes the metric-dual of \( p \in \CotBundle_q Q \).
	These dynamical equations have been already derived in~\parencite[Equation~(3.28)]{GayBalmazHolmRatiu2013}.
\end{example}

In the remainder of this section, we discuss the time dependence of the momentum map constraint.
For this purpose, consider the constraint map \( C: \CotBundle Q \times \LieA{g} \to \LieA{g}^* \) defined by
\begin{equation}
	C(q, p, \xi) = J(q, p) - \difpFrac{H}{\xi}(q, p, \xi).
\end{equation}
Clearly, the momentum map constraint is equivalent to \( C = 0 \).
\begin{prop}
	\label{prop:G-inv-H}
	Assume that \( H \) is \( G \)-invariant in the sense that
	\begin{equation}
		\label{eq:G-inv-H}
		H(g \cdot q, g \cdot p, \AdAction_g \xi) = H(q, p, \xi)
	\end{equation}
	for all \( g \in G \).
	Let \( t \mapsto \gamma(t) = \bigl(q(t), p(t), \xi(t) \bigr) \) be a curve in \( \CotBundle Q \times \LieA{g} \) satisfying the evolution equations~\eqref{eq:clebschLagrange:hamiltonian:hamiltonian}.
	Then, 
	\begin{equation}
		\difFrac{}{t} C\bigl(\gamma(t)\bigr) = - \CoadAction_{\xi(t)} C\bigl(\gamma(t)\bigl) - \difFrac{}{t} \left(\difpFrac{H}{\xi}\bigl(\gamma(t)\bigr) \right).
	\end{equation}
	In particular, a solution \( \gamma(t) \) of the evolution equations~\eqref{eq:clebschLagrange:hamiltonian:hamiltonian} with \( C\bigl(\gamma(t_0)\bigr) = 0 \) is tangent\footnote{That is, \( \tangent_{\gamma(t_0)} C \bigl(\dot{\gamma}(t_0)\bigr) = 0 \).} to the constraint set \( C = 0 \) at \( \gamma(t_0) \) if and only if 
	\begin{equation+}
		\difFracAt{}{t}{t_0} \left(\difpFrac{H}{\xi}\bigl(\gamma(t)\bigr)\right) = 0.
		\qedhere
	\end{equation+}
\end{prop}
\begin{proof}
	For clarity, we suppress the explicit time dependence in the subsequent calculation.
	By differentiating the \( G \)-invariance identity~\eqref{eq:G-inv-H} at the point \( (q, p, \xi) \) with respect to \( g \in G \), using~\eqref{eq:clebschLagrange:hamilton:decompCotBundle} and~\eqref{eq:clebschLagrange:hamilton:decompTHamilton}, we get
	\begin{equation}\begin{split}
		0 
		&= \dualPair*{\difpFrac{H}{q}}{\zeta \ldot q} + \dualPair*{\bar{K}(\zeta \ldot p)}{\difpFrac{H}{p}}{} + 
		\kappa\left(\difpFrac{H}{\xi}, \adAction_\zeta \xi\right)
	\end{split}\end{equation}
	for all \( \zeta \in \LieA{g} \).
	Here, \( \adAction \) denotes the adjoint action of \( \LieA{g} \), that is, \( \adAction_\zeta \xi = \commutator{\zeta}{\xi} \).
	Using this equation and the evolution equations~\eqref{eq:clebschLagrange:hamiltonian:hamiltonian}, we obtain
	\begin{equation}\begin{split}
		\difFrac{}{t} \kappa\bigl(C \bigl(\gamma(t)\bigr), \zeta\bigr)
			&= \difFrac{}{t} \left(\dualPair{p}{\zeta \ldot q} - \kappa\left(\difpFrac{H}{\xi}, \zeta\right)\right)
			\\
			&= \dualPair*{\DifFrac{}{t} p}{\zeta \ldot q} + \dualPair*{p}{\DifFrac{}{t} (\zeta \ldot q)} - \kappa\left(\difFrac{}{t} \difpFrac{H}{\xi}, \zeta\right)
			\\
			&= \dualPair{\bar{K}(\zeta \ldot p)}{\xi \ldot q} - \dualPair{\bar{K}(\xi \ldot p)}{\zeta \ldot q} + \kappa\left(\difpFrac{H}{\xi}, \adAction_\zeta \xi\right) 
			\\
			&\qquad + \dualPair{\bar{K}(\zeta \ldot p)}{\dot q} + \dualPair*{p}{\nabla_{\dot q} \zeta_*} - \kappa\left(\difFrac{}{t} \difpFrac{H}{\xi}, \zeta\right).
	\end{split}\end{equation}
	By~\eqref{eq:clebschLagrange:hamiltonian:proof:nablaXi}, the sum of the fourth and the fifth term vanishes.
	Moreover, by \( G \)-equivariance of \( J \) and by using~\eqref{eq:clebschLagrange:hamiltonian:proof:nablaXi}, we have
	\begin{equation}
		\dualPair{\bar{K}(\zeta \ldot p)}{\xi \ldot q} - \dualPair{\bar{K}(\xi \ldot p)}{\zeta \ldot q} = - \kappa(J(q, p), \adAction_\zeta \xi).
	\end{equation}
	This yields the assertion.
\end{proof}

A calculation similar to the one in the proof of~\cref{prop:G-inv-H} yields the following.
\begin{prop}
	\label{prop:clebschLagrange:lagrange:constraintConstantOfMotionEulerPoincare}
	Assume, instead, that \( H \) is \( G \)-invariant in the sense that
	\begin{equation}
		H(g \cdot q, g \cdot p, \xi) = H(q, p, \xi)\, ,
	\end{equation}
	for all \( g \in G \).
	Let \( t \mapsto \gamma(t) = \bigl(q(t), p(t), \xi(t) \bigr) \) be a curve in \( \CotBundle Q \times \LieA{g} \) satisfying the Clebsch--Hamilton equations~\eqref{eq:clebschLagrange:hamiltonian}.
	Then the curve \( t \mapsto J\bigl(\gamma(t)\bigr) \in \LieA{g}^* \) satisfies the Euler--Poincaré equation
	\begin{equation+}
		\difFrac{}{t} \left( J\bigl(\gamma(t)\bigr) \right) = - \CoadAction_{\xi(t)} \left( J\bigl(\gamma(t)\bigr) \right).
		\qedhere
	\end{equation+}
\end{prop}

\begin{remark}[Clebsch representation]
	Originally, the Clebsch representation refers to a special parameterization of the velocity field of an ideal incompressible fluid.
	\Textcite{MarsdenWeinstein1983} have shown that this classical example from fluid dynamics fits into the following geometric framework:
	Let \( (P, \poissonDot) \) be a Poisson manifold with Hamiltonian \( h: P \to \R \).
	A \emphDef{Clebsch representation} of \( P \) is a pair consisting of a symplectic manifold \( M \) and a Poisson map \( \psi: M \to P \).
	If we let \( H = h \circ \psi \), then \( \psi \) intertwines the Hamiltonian vector fields \( X_H \) and \( X_h \).
	Thus, by introducing possibly redundant variables, the Hamilton--Poisson equations on \( P \) are written in a symplectic Hamiltonian form on \( M \).
	The most important class of examples is obtained when \( P = \LieA{g}^* \) endowed with the Lie--Poisson bracket and when \( \psi = J \) is the momentum map for a symplectic \( G \)-action on \( M \).
	Then, \( J \) gives Clebsch variables in which the Euler--Poincaré equations on \( \LieA{g}^* \) are written in a symplectic Hamiltonian form.

	Now consider the Clebsch--Hamiltonian setup given by a Clebsch--Hamiltonian \( H: \CotBundle Q \times \LieA{g} \to \R \) which is \( G \)-invariant in the sense that \( H(g \cdot q, g \cdot p, \xi) = H(q, p, \xi) \).
	Then, \cref{prop:clebschLagrange:lagrange:constraintConstantOfMotionEulerPoincare} shows that \( \psi = J: \CotBundle Q \to \LieA{g}^* \) intertwines solutions of the Clebsch--Hamilton equations with solutions of the Euler--Poincaré equation on \( \LieA{g}^* \) parametrized by points \( (q,p) \in \CotBundle Q \).
	In other words, by introducing possibly redundant variables, the (generalized) Euler--Poincaré equations are written in a symplectic Clebsch--Hamilton form.

	Many equations (especially from hydrodynamics) can be written as Euler--Poincaré equations on some Lie algebra.
	For these cases, the Clebsch--Hamilton formalism thus provides a framework to construct different Clebsch-like representations.
	This observation might turn out to be especially advantageous for coupled equations, \eg Yang--Mills plasmas or relativistic fluids, for which the additional variables in \( \CotBundle Q \) admit a physical interpretation. 
\end{remark}

%%%%%%%%%%%%%%%%%%%%%%%%%%%%%%%%%%%%%%%%%%%%%%%%%%%%%%%%%%%%%%%%%%%%%%%%%%%%%%%%%%%%%%%%%%%%%%%%%%%%

\subsection{Relation to the standard Euler-Lagrange problem}
\label{Ham-Pic-gen}
In this section, we derive the Clebsch--Legendre transformation~\eqref{eq:clebschLagrange:clebschLegendre} from an ordinary Legendre transformation of an extended Lagrangian system.
This approach leads to a constraint analysis using the Dirac--Bergmann theory, which for the special cases of Yang--Mills theory and general relativity recovers the discussion of constraints in \parencite{BergveltDeKerf1986,BergveltDeKerf1986a,Giulini2015}.

%%%%%%%%%%%%%%%%%%%%%%%%%%%%%%%%%%%%%%%%%%%%%%%%%%%%%%%%%%%%%%%%%%%%%%%%%%%%%%%%%%%%%%%%%%%%%%%%%%%%%%

\subsubsection{The extended phase space}
One can view \( \xi \in \LieA{g} \) also as an ordinary configuration variable and thus formulate the Clebsch--Euler--Lagrange variational principle as an ordinary Euler--Lagrange problem.
For that purpose, we extend the configuration space to 
\begin{equation}
	Q_\ext \defeq Q \times \LieA{g} 
\end{equation}
and define
\begin{equation}
	\label{L-ext}
	L_\ext : \TBundle Q_\ext \to \R \, , \quad  L_\ext(q, \dot q, \xi, \dot \xi) = L(q, \dot q + \xi \ldot q, \xi) \, .
\end{equation}
 
The basic structure for the discussion of the Hamiltonian picture is the cotangent bundle
\begin{equation}
	\CotBundle Q_\ext = \CotBundle Q \times (\LieA{g} \times \LieA{g}^*)
\end{equation}
endowed with the natural product symplectic structure.
We usually denote points of \( \CotBundle Q_\ext \) by tuples \( (q, p, \xi, \nu) \).
First, we observe that the Lie group $\TBundle G$ acts naturally on $\TBundle Q_\ext$.
Under the right trivialization $\tau_\textrm{R}$ of $\TBundle G$,
\begin{equation}
	\tau_\textrm{R}: \LieA{g} \times G \to \TBundle G  \, , \quad (\zeta, g) \mapsto \zeta \ldot g \,, 
\end{equation}
 the group structure of \( \TBundle G \) is that of a semidirect product \( \LieA{g} \rSemiProduct_{\AdAction} G \).
 That is, the group multiplication is given by
\begin{equation}\label{eq:TGactionOnQext}
	(\xi, a) \cdot (\zeta, b) = (\xi + \AdAction_a \zeta, ab)
\end{equation}
for \( \xi, \zeta \in \LieA{g} \) and \( a,b \in G \).
The pair \( (\zeta, g) \in \LieA{g} \rSemiProduct_{\AdAction} G \) acts on \( Q_\ext \) by
\begin{equation}
	(\zeta, g) \cdot (q, \xi) \defeq (g \cdot q, \AdAction_g \xi + \zeta).
\end{equation}
A straightforward calculation shows that the natural lift of the \( \TBundle G \)-action to \( \CotBundle Q_\ext \) has the form
\begin{equation}\label{eq:extendedCotangentBundle:liftedAction}
	(\zeta, g) \cdot (q, p, \xi, \nu) = (g \cdot q, g \cdot p, \AdAction_g \xi + \zeta, \CoAdAction_g \nu).
\end{equation}
\begin{prop}\label{ext-CotB}
	Assume that the lifted \( G \)-action on the cotangent bundle \( \CotBundle Q \) has a momentum map \( J: \CotBundle Q \to \LieA{g}^* \) (which is automatic in finite dimensions).
	Then, the cotangent bundle \( \CotBundle Q_\ext \) carries the structure of a Hamiltonian \( \TBundle G \)-manifold with the equivariant momentum map \( J_\ext: \CotBundle Q_\ext \to \LieA{g}^* \times \LieA{g}^* \) given by 
	\begin{equation+}\label{eq:TGmomentumMap}
		J_\ext (q, p, \xi, \nu) = (\nu, \CoadAction_\xi \nu + J(q,p)).
		\qedhere
	\end{equation+}
\end{prop}
\begin{proof}
	Since the \( G \)-action on \( Q \) and the adjoint action are smooth, the \( \TBundle G \)-action on \( Q_\ext \) as defined in~\eqref{eq:TGactionOnQext} is smooth, too.
	The lifted \( \TBundle G \)-action on the cotangent bundle \( \CotBundle Q_\ext \) leaves the tautological \( 1 \)-form invariant and thus it is symplectic.
	Moreover, note that the fundamental vector field on \( Q_\ext \) generated by \( (\sigma, \varrho) \in \LieA{g} \times \LieA{g} \) is given by
	\begin{equation}
		(\sigma, \varrho) \ldot (q, \xi) = (\varrho \ldot q, \commutator{\varrho}{\xi} + \sigma).
	\end{equation}
	Using the defining relation~\eqref{eq:cotangentBundle:momentumMapDef} of the momentum map, we have
	\begin{equation}\begin{split}
		\kappa\left(J_\ext (q, p, \xi, \nu), (\sigma, \varrho) \right)
			&= \dualPair{(p, \nu)}{(\sigma, \varrho) \ldot (q, \xi)}
			\\
			&= \dualPair{p}{\varrho \ldot q} + \kappa(\nu, -\adAction_\xi \varrho + \sigma)
			\\
			&= \kappa(\nu, \sigma) + \kappa(\CoadAction_\xi \nu + J(q,p), \varrho),
	\end{split}\end{equation}
	which verifies the asserted formula~\eqref{eq:TGmomentumMap}.
	The equivariance of \( J_\ext \) is immediate from~\eqref{eq:cotangentBundle:momentumMapDef}.
\end{proof}

\subsubsection{The Dirac--Bergmann theory of constraints}
The extended Lagrangian \( L_\ext \) is never regular, because the velocity vector \( \dot \xi \) does not appear in \( L_\ext \).
In order to handle the non-degeneracy, we employ the Dirac--Bergmann theory of constraints.
Let us shortly review the algorithm in a geometric language as presented in \parencite{GotayNesterHinds1978}.
\begin{enumerate}
	\item
		Given a degenerate Lagrangian \( L: \TBundle Q \to \R \), let \( M_1 \subseteq \CotBundle Q \) be the image of the fiber derivative of \( L \). 
		Assume that \( M_1 \) is a smooth submanifold.
		Then, at least locally, \( M_1 \) is characterized by the vanishing of a collection of functions \( c_i \).
		The equations \( c_i(q, p) = 0 \) are called the \emphDef{primary constraints}.
	\item
		The Hamiltonian \( H \), as a smooth function on \( M_1 \), is defined by the Legendre transformation
		\begin{equation}
			H\left(q, \difpFrac{L}{\dot{q}}(q, \dot q)\right) 
				\defeq \dualPair*{\difpFrac{L}{\dot{q}}(q, \dot q)}{\dot q} - L(q, \dot q).
		\end{equation}
		Let \( M_2 \subseteq M_1 \) be the subset characterized by the \emphDef{secondary constraints} as follows:
		\begin{equation}
			\label{eq:clebsch:hamiltonian:secondaryConstraint}
			m \in M_2 \text{ if and only if } \dualPair{\dif_m H}{\TBundle_m M_1 \intersect (\TBundle_m M_1)^\omega} = 0,
		\end{equation}
		where \( (\TBundle_m M_1)^\omega \) denotes the symplectic orthogonal of \( \TBundle_m M_1 \) in \( \TBundle_m (\CotBundle Q) \) with respect to the canonical symplectic form \( \omega \).
		Assume that \( M_2 \) is a smooth submanifold of \( M_1 \).
	\item
		Iterate the process to get a chain of submanifolds
		\begin{equationcd}
			\ldots \to[r]
				& M_i \to[r]
				& \ldots \to[r]
				& M_3 \to[r]
				& M_2 \to[r]
				& M_1 \to[r]
				& \CotBundle Q, \nonumber
		\end{equationcd}
		defined by the condition that \( m \in M_i \) is a point of \( M_{i+1} \) if and only if \( \dualPair{\dif_m H}{\TBundle_m M_1 \intersect (\TBundle_m M_i)^\omega} = 0 \).
		If the algorithm happens to terminate at a non-empty submanifold \( M_f \), then the dynamics is, by construction, tangent to \( M_f \) and is Hamiltonian with respect to the Hamiltonian \( \restr{H}{M_f} \).
\end{enumerate}
\begin{remark}
	\Textcite{GotayNesterHinds1978} discuss this algorithm in an infinite-dimensional Banach setting but under the rather restrictive assumption that the symplectic form is strongly symplectic.
	This assumption never holds for symplectic Fréchet manifolds.
	Moreover, in \parencite{GotayNesterHinds1978} it is assumed that all subsets \( M_i \) are submanifolds of \( \CotBundle Q \).
	In our setting, \( M_2 \) turns out to be a momentum map level set and thus it is not a smooth manifold unless the action is free.
	In view of these problems, we take the Dirac--Bergmann algorithm only as a guide and verify directly at the end that the Hamiltonian system so obtained is indeed equivalent to the Clebsch--Hamilton equations, which in turn are equivalent to the degenerate Lagrangian system we started with, see \cref{prop:clebschLagrange:legendre}.
\end{remark}
We now apply this algorithm to the degenerate Lagrangian $L_\ext: \TBundle Q_\ext \to \R $ defined in~\eqref{L-ext}.
The fiber derivative of \( L_\ext \) is given by
\begin{equation}
	\left(\difpFrac{L_\ext}{\dot{q}}, \difpFrac{L_\ext}{\dot{\xi}}\right)(q, \dot q, \xi, \dot \xi)
		= \left(\difpFrac{L}{\dot{q}}(q, \dot q + \xi \ldot q, \xi), 0\right).
\end{equation}
Note that this Legendre transformation of the extended system is (essentially) equivalent to the Clebsch--Legendre transformation introduced in~\eqref{eq:clebschLagrange:clebschLegendre}.
In the case when \( L \) is regular, the range of the fiber derivative is thus 
\begin{equation}
 	M_1 \equiv \CotBundle Q \times \LieA{g} \times \set{0} \subseteq \CotBundle Q_\ext = \CotBundle Q \times \LieA{g} \times \LieA{g}^* \, .
\end{equation}
In other words, the first primary constraint entails that the canonical momentum conjugate to \( \xi \) has to vanish, \ie, \( \nu = 0 \).
This should not come as a surprise as \( L_\ext \) does not depend on the velocity vector \( \dot \xi \). 
The Hamiltonian \( H_\ext \) on \( \CotBundle Q \times \LieA{g} \) is defined by the Legendre transformation of \( L_\ext \),
\begin{equation}
\label{H-ext}
	H_\ext(q,p,\xi) \defeq \dualPair*{p}{\dot q} - L_\ext(q, \dot q, \xi, 0) \, ,
\end{equation}
where the relation
\begin{equation}
	p = \difpFrac{L_\ext}{\dot{q}}(q, \dot q, \xi, \dot \xi)
		= \difpFrac{L}{\dot{q}}(q, \dot q + \xi \ldot q, \xi)
\end{equation}
determines $\dot q$ as a function of $q,p$ and $\xi$.
As the next step of the constraint analysis, we have to determine the set of points \( (q, p, \xi) \in M_1 \) for which the linearized Hamiltonian vanishes on the symplectic orthogonal of \( \TBundle_{(q,p,\xi)} M_1 \).
We clearly have
\begin{equation}
	\bigl(\TBundle_{(q,p,\xi)} M_1\bigr)^\omega
		= \bigl(\TBundle_{(q,p)} (\CotBundle Q) \times \LieA{g} \times \set{0}\bigr)^\omega
		= \set{0} \times (\LieA{g} \times \set{0})
		\subseteq \TBundle_{(q,p)} (\CotBundle Q) \times (\LieA{g} \times \LieA{g}^*),
\end{equation}
because the canonical symplectic form on \( \CotBundle Q \) is weakly non-degenerate and \( \LieA{g} \times \set{0} \) is a Lagrangian subspace of \( \LieA{g} \times \LieA{g}^* \).
Thus, the secondary constraint entails that the derivative of \( H_\ext \) in the \( \xi \)-direction has to vanish, \ie
\begin{equation}
\label{constr-H-ext}
	\difpFrac{H_\ext}{\xi}(q,p,\xi) = 0.
\end{equation}
To summarize, after this step of the constraint analysis, we obtain the following system of constrained Hamilton equations on \( \CotBundle Q_\ext \):
\begin{subequations}\label{eq:clebschLagrange:hamiltonianExt}\begin{gather}
	\label{eq:clebschLagrange:hamiltonianExt:hamiltonian}
	\difpFrac{H_\ext}{q}(q, p, \xi) = - \DifFrac{}{t} p, \qquad \difpFrac{H_\ext}{p}(q, p, \xi) = \dot q,
	\\
	\label{eq:clebschLagrange:hamiltonianExt:constraints}
	\nu = 0, \qquad \difpFrac{H_\ext}{\xi}(q,p,\xi) = 0.
\end{gather}\end{subequations}
The constraints~\eqref{eq:clebschLagrange:hamiltonianExt:constraints} are, in general,  not preserved in time and, thus, the Dirac--Bergmann algorithm may yield further constraints.
We will return to this question at the end of the section, but first let us relate the constrained Hamiltonian system~\eqref{eq:clebschLagrange:hamiltonianExt} to the Clebsch--Hamiltonian system~\eqref{eq:clebschLagrange:hamiltonian}.
For this, we need the following.
\begin{lemma}
	The following identity holds:
	\begin{equation+}\label{eq:compositeHamiltonian}
		H_\ext(q, p, \xi) = H(q, p, \xi) - \kappa(J(q, p), \xi),
	\end{equation+}
	where \( H \) is the Clebsch--Hamiltonian obtained via the Clebsch--Legendre transformation, \cf~\eqref{eq:hamiltonian}.
\end{lemma}
\begin{proof}
Using~\eqref{H-ext}, the momentum map relation and the definition of \( H \) as the Legendre transform of \( L \) (with \( v = \dot q + \xi \ldot q \) in the above notation), we calculate
\begin{align}
	H_\ext(q, p, \xi)
		&= \dualPair{p}{\dot q} - L_\ext(q, \dot q, \xi, 0) \nonumber
		\\
		&= \dualPair*{p}{\dot q} - L(q, \dot q + \xi \ldot q, \xi)
		\\
		&= \dualPair*{p}{\dot q + \xi \ldot q} - L(q, \dot q + \xi \ldot q, \xi) - \dualPair*{p}{\xi \ldot q}\nonumber
		\\
		&= H \left(q, p, \xi \right) - \kappa(J(q, p), \xi).\nonumber \qedhere
\end{align}
\end{proof}
The identity~\eqref{eq:compositeHamiltonian} allows us to reformulate the equations~\eqref{eq:clebschLagrange:hamiltonianExt} in a form which involves \( H \) only.
Using \( \kappa(J(q, p),\xi) = \dualPair{p}{\xi \ldot q} \), we obtain
\begin{equation}\begin{split}
	\difpFrac{H_\ext}{q}(q, p, \xi) 
		&= \difpFrac{H}{q}(q, p, \xi) + \bar{K}(\xi \ldot p),
	\\
	\difpFrac{H_\ext}{p}(q, p, \xi)
		&= \difpFrac{H}{p}(q, p, \xi) - \xi \ldot q,
\end{split}\end{equation}
so that the evolution equations~\eqref{eq:clebschLagrange:hamiltonianExt:hamiltonian} take the form
\begin{equation}
	\difpFrac{H}{q}(q, p, \xi) = - \DifFrac{}{t} p - \bar{K}(\xi \ldot p), \qquad \difpFrac{H}{p}(q, p, \xi) = \dot q + \xi \ldot q.
\end{equation}
Similarly, for the \( \xi \)-derivative of \( H_\ext \), using~\eqref{eq:compositeHamiltonian}, we find
\begin{equation}
	\difpFrac{H_\ext}{\xi}(q, p, \xi)
		= \difpFrac{H}{\xi}(q, p, \xi) - J(q, p).
\end{equation}
The constraints~\eqref{eq:clebschLagrange:hamiltonianExt:constraints} are, therefore, equivalent to
\begin{equation}
	\label{eq:Constr}
	\nu = 0, \qquad J(q, p) = \difpFrac{H}{\xi}(q, p, \xi).
\end{equation}
In summary, the constrained Hamiltonian system~\eqref{eq:clebschLagrange:hamiltonianExt} obtained by the Dirac-Bergmann analysis of the extended system is equivalent to the Clebsch--Hamilton system~\eqref{eq:clebschLagrange:hamiltonian}.

\begin{remark}
	\label{prop:clebschLagrange:constraintsAsMomentumMapExtConstraint}
	Using \cref{ext-CotB}, we see that the constraints defined by~\eqref{eq:Constr} are equivalent to the momentum map constraint
	\begin{equation+}
		J_\ext(q, p, \xi, \nu) = \left(0, \difpFrac{H}{\xi}(q, p, \xi)\right).
		\qedhere	
	\end{equation+}
\end{remark}
In \cref{prop:G-inv-H}, we have seen that the dynamics is tangent to the momentum map constraint if \( H \) is \( G \)-invariant in the sense of~\eqref{eq:G-inv-H} and if, additionally, it does not explicitly depend on the \( \xi \)-variable.
Thus, in this case, the Dirac--Bergmann algorithm terminates at this stage without introducing further constraints.
Let us spell out the details.
Since \( H \) does not explicitly depend on \( \xi \), we have \( M_2 = J^{-1}(0) \times \LieA{g} \).
Note that \( M_2 \) might have singularities if the action is not free.
The \enquote{tangent space} to \( J^{-1}(0) \) at a point \( (q, p) \) is the kernel of \( \tangent_{q, p} J \), which by the bifurcation lemma equals \( {(\LieA{g} \ldot (q, p))}^\omega \), \cf \parencite[Proposition~4.5.14]{OrtegaRatiu2003} for the finite-dimensional case and \parencite{DiezThesis} for the infinite-dimensional setting.
Thus, for a \( G \)-invariant Clebsch--Hamiltonian \( H \), the condition\footnote{Here, we have used the identity \( (\LieA{g} \ldot (q, p))^{\omega \omega} = \LieA{g} \ldot (q, p) \), which in infinite dimensions requires additional assumptions of functional-analytic nature.
For a precise formulation of the bifurcation lemma in an infinite-dimensional context see \parencite{DiezThesis}.}
\begin{equation}
	0 
		= \dualPair{\dif_{q, p, \xi} H_\ext}{\TBundle_{q, p, \xi} M_1 \intersect (\TBundle_{q, p, \xi} M_2)^\omega}
		= \dualPair{\dif_{q, p} H}{\LieA{g} \ldot (q,p)},
\end{equation}
is automatically satisfied for all \( (q, p, \xi) \in \CotBundle Q \times \LieA{g} \) and the Dirac algorithm terminates at \( M_2 \).

%%%%%%%%%%%%%%%%%%%%%%%%%%%%%%%%%%%%%%%%%%%%%%%%%%%%%%%%%%%%%%%%%%%%%%%%%%%%

\subsubsection{Reduction by stages}
\label{sec:clebschLagrange:reductionStages}

%%%%%%%%%%%%%%%%%%%%%%%%%%%%%%%%%%%%%%%%%%%%%%%%%%%%%%%%%%%%%%%%%%%%%%%%%%%%%

As we have seen in \cref{prop:clebschLagrange:constraintsAsMomentumMapExtConstraint}, the constraints in the Dirac--Bergmann algorithm can be implemented in terms of the momentum map \( J_\ext \) for the lifted \( \TBundle G \)-action on \( \CotBundle Q_\ext \).
Thus it is natural to think of the symplectically reduced space \( \CotBundle Q_\ext \sslash \TBundle G \) as the true phase space of the theory.
In the sequel, we use the theory of symplectic reduction by stages \parencite{MarsdenMisiolekEtAl2007} for the tangent group \( \TBundle G \) to realize the two steps in the Dirac--Bergmann process as two separate symmetry reductions.

The starting point is the cotangent bundle \( \CotBundle Q_\ext \) of the extended configuration space \( Q_\ext \).
Let us assume that the Clebsch--Hamiltonian $H$ is $G$-invariant in the sense of~\eqref{eq:G-inv-H} and that it does not explicitly depend on \( \xi \).
Then, the constraints~\eqref{eq:Constr} have the form \( \nu = 0 \) and \( J(q, p) = 0 \), or in other words, \( J_\ext(q,p,\xi,\nu) = 0 \).
%Moreover, \cref{prop:G-inv-H} shows that these constraints are preserved by the dynamics. 
As we have seen above, the right trivialization \( \tau_\textrm{R} \) identifies \( \TBundle G \) with the semidirect product \( \LieA{g} \rSemiProduct_{\AdAction} G \).
It is hence natural to perform symplectic reduction by stages: first quotient out by the Lie algebra \( \LieA{g} \) and then by the Lie group \( G \).

Due to the particularly simple action of \( \LieA{g} \) on \( Q_\ext \) the reduced phase space at the first stage is symplectomorphic to \( \CotBundle Q \).
Indeed, the momentum map for the \( \LieA{g} \)-action on \( \CotBundle Q_\ext \) is simply given by
\begin{equation}
	J_{\LieA{g}} (q, p, \xi, \nu) = \nu.
\end{equation}
Hence, the condition \( J_{\LieA{g}} = 0 \) cuts out exactly the first constraint submanifold \( M_1 = \CotBundle Q \times \LieA{g} \times \set{0} \subset \CotBundle Q_\ext \).
Moreover, by~\eqref{eq:extendedCotangentBundle:liftedAction}, \( \LieA{g} \) acts on \( \CotBundle Q_\ext \) by translation in the \( \LieA{g} \)-factor and thus the quotient \( J_{\LieA{g}}^{-1}(0) \slash \LieA{g} \) is diffeomorphic to \( \CotBundle Q \).
The regular cotangent bundle reduction theorem \parencite[Theorem~6.6.1]{OrtegaRatiu2003} shows that the reduced symplectic form coincides with the canonical one (the problems coming from the infinite-dimensional setting can easily be handled, because the \( \LieA{g} \)-bundle \( Q_\ext \to Q \) is trivial).
By~\parencite[Lemma~4.2.6]{MarsdenMisiolekEtAl2007}, the momentum map for the residual \( G \)-action on the reduced space \( \CotBundle Q \) is induced by the map
\begin{equation}
	\pr_2 \circ \restr{(J_\ext)}{J_{\LieA{g}}^{-1}(0)} (q, p, \xi) = J(q, p).
\end{equation}
That is, it coincides with the momentum map for the lifted action on \( \CotBundle Q \).
This can be also directly deduced by noting that the residual \( G \)-action coincides with the lifted action on the cotangent bundle.
In particular, the momentum level set \( J^{-1}(0) \) coincides with the secondary constraint set \( M_2 \).

As the momentum maps under consideration are equivariant, the reduction by stages theorem for semidirect product actions~\parencite[Theorem~4.2.2]{MarsdenMisiolekEtAl2007} implies that the two-stage reduced space is symplectomorphic to the all-at-once reduced space\footnote{To be more precise, the reduction by stages theorem~\parencite[Theorem~4.2.2]{MarsdenMisiolekEtAl2007} is formulated for the free and proper action of a finite-dimensional group on a finite-dimensional phase space. Nonetheless, the peculiarities of our infinite-dimensional setting can be handled easily due to the simple form of the first reduction.}.
We thus have shown the following.
\begin{prop}
	\label{prop:clebschLagrange:reductionStages}
  The symplectically reduced space \( \check{M} = J^{-1}(0) \slash G \equiv \CotBundle Q \sslash G \) is symplectomorphic to the symplectic quotient \( J_\ext^{-1}(0) \slash \TBundle G \equiv \CotBundle Q_\ext \sslash \TBundle G \) and we have the following reduction by stages diagram:
	\begin{equation}\begin{tikzcd}[column sep=5.9em, row sep=0.02em]
		\CotBundle Q_\ext
			\ar[r, twoheadrightarrow, "{\sslash \, \LieA{g}}"]
			\ar[rr, twoheadrightarrow, "{\sslash \, \TBundle G}", swap, bend right]
		& \CotBundle Q
			\ar[r,twoheadrightarrow, "{\sslash \, G}"]
		& \check{M}.
	\end{tikzcd}\end{equation} 
\end{prop}
By this proposition, the symplectic reduction procedure boils down to the symplectic reduction of $\CotBundle Q$ with respect to $G$.
This singular reduction will be discussed in detail elsewhere \parencite{DiezRudolphReduction}.  

Finally, let us comment on the reduction of dynamics.
We continue to work in the setting where the Clebsch--Hamiltonian \( H \) is \( G \)-invariant in the sense of~\eqref{eq:G-inv-H} and does not explicitly depend on \( \xi \).
In particular, we view \( H \) as a smooth function on \( \CotBundle Q \).
Recall that the Hamiltonian $H_\ext$ is only defined on \( J_{\LieA{g}}^{-1}(0) \) and not on the whole of \( \CotBundle Q_\ext \).
By~\eqref{eq:extendedCotangentBundle:liftedAction}, the $\LieA{g}$-action on \( \CotBundle Q_\ext \) has the form
\begin{equation}
	(\zeta, {\mathbbm 1}) \cdot (q, p, \xi, \nu) = (q, p, \xi + \zeta, \nu) 
\end{equation}
and, thus, $H_\ext$ as given by~\eqref{eq:compositeHamiltonian}  is \emph{not} \( \LieA{g} \)-invariant.
That is, it does not descend to the first reduced space \( J_{\LieA{g}}^{-1}(0) \slash \LieA{g} = \CotBundle Q \).
On the other hand, if also the secondary constraint \( J = 0 \) is imposed, then \( H_\ext \) coincides with $H$ and is \( \TBundle G \)-invariant.
Hence, the Hamiltonian \( H_\ext \) \emph{does} descend to the completely reduced space \( \check{M} \isomorph \CotBundle Q_\ext \sslash \TBundle G \) and the reduction of dynamics boils down to the reduction of $H: \CotBundle Q \to \R$ with respect to $G$.
Let us summarize.
\begin{thm}
	\label{prop:clebschLagrange:equivalenceHamiltonianReductions}
	Let \( H: \CotBundle Q \to \R \) be a \( \xi \)-independent Clebsch--Hamiltonian, which is \( G \)-invariant in the sense of~\eqref{eq:G-inv-H}.
	Assume that the lifted \( G \)-action to \( \CotBundle Q \) has a momentum map \( J \).
	Then, the following systems of equations are equivalent:
	\begin{thmenumerate}
		\item
			The Clebsch--Hamilton equations~\eqref{eq:clebschLagrange:hamiltonian} on \( \CotBundle Q \times \LieA{g} \) with respect to \( H \).
		\item
			The constraint Hamilton equations~\eqref{eq:clebschLagrange:hamiltonianExt} on \( \CotBundle Q_\ext \) with respect to the Hamiltonian \( H_\ext: \CotBundle Q \times \LieA{g} \to \R \) defined by~\eqref{eq:compositeHamiltonian}.
		\item
			The Hamilton equations on \( \check{M} = \CotBundle Q \sslash G \isomorph \CotBundle Q_\ext \sslash \TBundle G  \) with respect to the reduced Hamiltonian \( \check{H}: \check{M} \to \R \) defined by
			\begin{equation}
				\pi^* \check{H} = \restr{H}{J^{-1}(0)},
			\end{equation}
			where \( \pi: J^{-1}(0) \to J^{-1}(0) \slash G = \check{M} \) is the natural projection.
			\qedhere 
	\end{thmenumerate}
\end{thm}
Since the Clebsch--Hamiltonian \( H: \CotBundle Q \times \LieA{g} \to \R \) is \( \xi \)-independent, it may be viewed as an ordinary Hamiltonian on \( \CotBundle Q \).
We emphasize, however, that the ordinary Hamilton equation on \( \CotBundle Q \) with respect to \( H \) are \emph{not} equivalent to the Clebsch--Hamilton equations on \( \CotBundle Q \times \LieA{g} \).
This explains to some extend the problem one faces when trying to quantize relativistic field theories: in order to make sense of the dynamics, the Hamiltonian and the momentum constraint have to be quantized simultaneously, or, alternatively, the symmetry has to be reduced completely before quantization.

%%%%%%%%%%%%%%%%%%%%%%%%%%%%%%%%%%%%%%%%%%%%%%%%%%%%%%%%%%%%%%%%%%%%%%%%%%%%%%%%%%%%%%%%

\section{Yang--Mills--Higgs theory}
\label{sec:yangMillsHiggs}

%%%%%%%%%%%%%%%%%%%%%%%%%%%%%%%%%%%%%%%%%%%%%%%%%%%%%%%%%%%%%%%%%%%%%%%%%%%%%%%%%%%%%%%%

In this section, we will show how the Yang--Mills--Higgs system fits into the general Clebsch--Lagrange variational framework discussed in \cref{sec:clebschLagrange}.

Let \( (M, \eta) \) be a \( 4 \)-dimensional oriented Lorentzian manifold with signature \( (- + +\, +) \).
The underlying geometry of a Yang--Mills-Higgs field, is that of a principal \( G \)-bundle \( P \to M \), where \( G \) is a connected compact Lie group.
In order to establish the notation, we recall the geometric picture and we refer for details to \parencite{RudolphSchmidt2014}.
A connection in \( P \) is a splitting of the tangent bundle \( \TBundle P = \VBundle P \oplus \HBundle P \) into the canonical vertical distribution \( \VBundle P \) and a horizontal distribution \( \HBundle P \).
Recall that \( \VBundle P \) is spanned by the Killing vector fields \( p \mapsto \xi \ldot p \) for \( \xi \in \LieA{g} \).
Equivalently, a connection is given by a \( G \)-equivariant \( 1 \)-form \( A \in \DiffFormSpace^1(P, \LieA{g}) \).
A bosonic matter field is a section \( \varphi \) of the associated vector bundle \( F = P \times_G \FibreBundleModel{F} \), where the typical fiber \( \FibreBundleModel{F} \) carries a \( G \)-representation.
Thus, the space of configurations of Yang--Mills--Higgs theory consists of pairs \( (A, \varphi) \).
It is obviously the product of the infinite-dimensional affine space \( \ConnSpace \) of connections and the space of sections \( \SectionSpaceAbb{F} \) of \( F \).
On \( \ConnSpace \times \SectionSpaceAbb{F} \) we have a left action of the group \( \GauGroup = \sSectionSpace(P \times_G G) \) of local gauge transformations, 
\begin{equation}
\label{LocGTr}
	A \mapsto \AdAction_{\lambda} A - \dif \lambda \, \lambda^{-1},
	\quad
	\varphi \mapsto \lambda \cdot \varphi,
\end{equation}
for \( \lambda \in \GauGroup \).
Next, recall the notion  of the covariant exterior derivative, which we denote by \( \dif_A \).
Let \( \alpha \in \DiffFormSpace^k(M, F) \) and let \( \tilde{\alpha} \in \DiffFormSpace^k(P, \FibreBundleModel{F}) \) be its associated horizontal form\footnote{This one-to-one correspondence will be used throughout the text without further notice.}.
Then, 
\begin{equation}
	\dif_A \tilde{\alpha} = \dif \tilde{\alpha} + A \wedgeldot \tilde{\alpha},
\end{equation}
where \( \wedgeldot: \DiffFormSpace^r(P, \LieA{g}) \times \DiffFormSpace^k(P, \FibreBundleModel{F}) \to \DiffFormSpace^{r+k}(P, \FibreBundleModel{F}) \) is the natural operation obtained by combining the Lie algebra action \( \LieA{g} \times \FibreBundleModel{F} \to \FibreBundleModel{F}, (\xi, f) \mapsto \xi \ldot f \) with the wedge product operation.
An important special case is provided by the curvature \( F_A = \dif_A A \) of the connection \( A \), which is a horizontal \( 2 \)-form of type \( \AdAction \) on \( P \) or, equivalently, a \( 2 \)-form on \( M \) with values in the adjoint bundle \( \AdBundle P = P \times_G \LieA{g} \).
Next, note that the Lie algebra $\GauAlgebra$ of $\GauGroup$ may be naturally identified with \(  \sSectionSpace(\AdBundle P ) \).
The infinitesimal action of $\xi \in \GauAlgebra$ on \( \ConnSpace(P) \) is given by
\begin{equation}\label{eq:generalGauge:infinitisimalGaugeAction}
	\xi \ldot A = \adAction_\xi A - \dif \xi = - \dif_A \xi.
\end{equation}

Now we can formulate the variational principle for the Yang--Mills--Higgs system on \( \ConnSpace \times \SectionSpaceAbb{F} \).
For that purpose we fix an \( \AdAction_G \)-invariant scalar product on \( \LieA{g} \) and a \( G \)-invariant scalar product on \( \FibreBundleModel{F} \).
The Lagrangian for this model is given by the following top differential form on $M$:
\begin{equation}\label{eq:yangMillsHiggs:action}
	\SectionSpaceAbb{L}_\textrm{YMH}(A, \varphi) = \frac{1}{2} \wedgeDual{F_A}{\hodgeStar F_A} + \frac{1}{2} \wedgeDual{\dif_A \varphi}{\hodgeStar \dif_A \varphi} - V(\varphi) \vol_\eta,
\end{equation}
where \( V: F \to \R \) denotes the Higgs potential induced from a smooth \( G \)-invariant function \( \FibreBundleModel{V}: \FibreBundleModel{F} \to \R \).
In order to underline that the Hodge dual is defined in terms of a linear functional on the space of differential forms, we use the convention that the Hodge dual of a vector-valued differential form \( \alpha \in \DiffFormSpace^k(M, F) \) is the \emph{dual-valued}\footnotemark{} differential form \( \hodgeStar \alpha \in \DiffFormSpace^{4-k}(M, F^*) \).
\footnotetext{Although this convention is a bit non-standard, the consistent use of dual-valued forms has the advantage that dual objects are clearly marked, which will be helpful later to identify them as points in the cotangent bundle (as opposed to elements of the tangent bundle).}%
As we are using \( G \)-invariant scalar products on \( \LieA{g} \) and \( \FibreBundleModel{F} \), we have the equivariance property
\begin{equation}\label{equivariance:hodgeStar}
	\hodgeStar \, (\AdAction_g \alpha) = \CoAdAction_g (\hodgeStar \alpha).
\end{equation}
Moreover, for \( \alpha \in \DiffFormSpace^k(M, F) \) and \( \beta \in \DiffFormSpace^{4-k}(M, F^*) \), we denote by \( \wedgeDual{\alpha}{\beta} \) the real-valued top-form that arises from combining the wedge product with the natural pairing \( \dualPairDot: F \times F^* \to \R \).
By~\eqref{equivariance:hodgeStar}, the Lagrangian \( L \) is gauge invariant.

The Euler--Lagrange equations corresponding to the Lagrangian~\eqref{eq:yangMillsHiggs:action}, called the \emphDef{Yang--Mills--Higgs equations}, read as follows:
\begin{subequations}\label{eq:yangMillsHiggs4d}\begin{align}
	\dif_A \hodgeStar F_A + \varphi \diamond \hodgeStar \dif_A \varphi = 0,
	\label{eq:yangMillsHiggs:4d:ym}
	\\
	\dif_A \hodgeStar \dif_A \varphi + V' (\varphi) \vol_\eta = 0,
	\label{eq:yangMillsHiggs:4d:higgs}
\end{align}\end{subequations}
where the derivative \( V' (\varphi) \) of \( V \) at the point \( \varphi \) is viewed as a fiberwise linear functional on \( F \) and the diamond product\footnote{The diamond operator often occurs in the study of Lie--Poisson systems and this is were we borrowed the notation from. Note that in a purely algebraic setting, the diamond product boils down to a map \( \diamond: F \times F^* \to \LieA{g}^* \) dual to the Lie algebra action.
Hence, it is a momentum map for the \( G \)-action on \( \CotBundle F \).
For example, if we consider the action of \( G = \SOGroup(3) \) on \( F = \R^3 \) and identify \( F^* \) with \( \R^3 \), then the diamond product becomes the classical cross product.}
\begin{equation}
	\diamond: \DiffFormSpace^k(M, F) \times \DiffFormSpace^{\dim M-r-k}(M, F^*) \to \DiffFormSpace^{\dim M-r}(M, \CoAdBundle P)
\end{equation}
is defined by
\begin{equation}
	\label{eq:yangMillsHiggs:defDiamond}
	\wedgeDual{\xi}{(\alpha \diamond \beta)}
		= \wedgeDual{(\xi \wedgeldot \alpha)}{\beta} 
		\in \DiffFormSpace^{\dim M}(M) 
\end{equation}
for all \( \xi \in \DiffFormSpace^r(M, \AdBundle P) \).
The diamond product is equivariant with respect to gauge transformations in the sense that
\begin{equation}
	(\lambda \cdot \alpha) \diamond (\lambda \cdot \beta) = \CoAdAction_\lambda (\alpha \diamond \beta)
\end{equation}
holds for all \( \lambda \in \GauGroup(P) \).
Indeed, we have
\begin{equation}\begin{split}
	\wedgeDual{\xi}{((\lambda \cdot \alpha) \diamond (\lambda \cdot \beta))}
		&= \wedgeDual{(\xi \wedgeldot (\lambda \cdot \alpha))}{(\lambda \cdot \beta)}
		\\
		&= \wedgeDual{\lambda \cdot ((\AdAction_\lambda^{-1}\xi) \wedgeldot \alpha)}{(\lambda \cdot \beta)}
		\\
		&= \wedgeDual{\AdAction_\lambda^{-1}\xi}{(\alpha \diamond \beta)}
		\\
		&= \wedgeDual{\xi}{\CoAdAction_\lambda(\alpha \diamond \beta)} \, ,
\end{split}\end{equation}
for all \( \xi \in \DiffFormSpace^r(M, \AdBundle P) \).

In order to write the Yang--Mills equation~\eqref{eq:yangMillsHiggs4d} as a dynamical system we have to single-out a time direction and to split the equations in time and space directions.
This decomposition is standard in the physics literature and pretty straightforward in local coordinates.
A derivation of the following identities using a geometric, coordinate-independent language can be found in \cref{sec:yangMillsHiggs:decomposition}.
Let us assume that the spacetime \( (M, \eta) \) is globally hyperbolic and time-oriented.
Then, there exists a $3$-dimensional manifold $\Sigma$ and  a diffeomorphism \( \iota: \R \times \Sigma \to M \) such that
\begin{equation}
	\iota^* \eta = - \ell(t)^2 \dif t^2 + g(t),
\end{equation}
where \( \ell \) is a smooth time-dependent positive function on \( \Sigma \) and \( g \) is a time-dependent Riemannian metric on \( \Sigma \), see \parencite[Theorem~1.3.10]{BarGinouxEtAl2007}.
The Hodge operator associated to \( g \) is denoted by \( \hodgeStar_g \).
Moreover, let \( \diamond_\Sigma \) denote the diamond product relative to \( \Sigma \).
For every \( t \in \R \), we denote by \( \iota_t: \Sigma \to M \) the induced embedding and by \( \Sigma_t \) the image of \( \Sigma \) under \( \iota_t \).
The submanifold $\Sigma_t$ is a Cauchy hypersurface at $t$. 
In what follows, we assume $\Sigma$ to be compact and without boundary.
This assumption amounts to requiring that the fields satisfy suitable boundary condition at spacial infinity. 
Following \parencite[Section~6A]{GotayIsenbergMarsden2004}, a \emphDef{slicing} of \( P \) over \( \iota \) is a principal \( G \)-bundle \( \vec{\pi}: \vec{P} \to \Sigma \) and a principal bundle isomorphism \( \hat{\iota} \) fitting into the following diagram:
\begin{equationcd}
	\R \times \vec{P} \to[r, "\hat{\iota}"] \to[d, "\id_\R \times \vec{\pi}", swap]
		& P \to[d, "\pi"]
	\\
	\R \times \Sigma \to[r, "\iota"]
		& M.
\end{equationcd}
By \parencite[Corollary~4.9.7]{Husemoller1966}, such a slicing always exists for small times; which suffices for the study of the initial value problem.
In order to simplify the presentation, let us assume that the slicing exists for all \( t \in \R \).

Relative to the splitting \( \hat{\iota}: \R \times \vec{P} \to P \), the objects living on \( P \) decompose into time-dependents objects living on \( \vec{P} \) and objects normal to \( \vec{P} \).
For example, a function \( f: P \to \R \) on \( P \) yields a time-dependent function \( \vec{f}(t) = f \circ \hat{\iota}(t, \cdot) \) on \( \vec{P} \).
In the sequel, we will usually suppress the diffeomorphisms \( \hat{\iota} \) and \( \iota \) in our notation.
Accordingly, a vector field on \( M \) can be written as \( X = X^0 \, \difp_t + \vec{X} \), where \( X^0(t) \in \sFunctionSpace(\Sigma) \) and \( \vec{X}(t) \in \sSectionSpace(\TBundle \Sigma ) \) for every \( t \).
A connection \( A \) in \( P \) decomposes as
\begin{equation}
	A = A_0 \dif t + \vec{A}
\end{equation}
into \( A_0(t) \in \sSectionSpace(\AdBundle \vec{P}) \) and a time-dependent connection \( \vec{A}(t) \) in \( \vec{P} \).
Moreover, the decomposition of the curvature takes the form
\begin{equation}
 	F_A = E \wedge \dif t + B,
\end{equation}
where, as usual, we have introduced the time-dependent color-electric field \( E \defeq \dif_{\vec{A}} \, A_0 - \dot{\vec{A}} \) and the color-magnetic field \( B \defeq F_{\vec{A}} \).
A section \( \varphi \in \sSectionSpace(F) \) is the same as a time-dependent section of \( \vec{F} \) and its covariant differential reads
\begin{equation}
\dif_A \varphi
		= \left(\partial_t^{A_0}\varphi\right) \dif t + \dif_{\vec{A}} \varphi,  	
\end{equation}
where \( \partial_t^{A_0} \varphi \defeq \dot{\varphi} + A_0 \ldot \varphi \) is the covariant time-derivative.
According to \cref{prop:splitting:yangMillsHiggs}, the Yang--Mills--Higgs equations~\eqref{eq:yangMillsHiggs4d} take the following form:
\begin{subequations}\label{eq:yangMillsHiggs3dEvol}\begin{align+}
	\dif_{\vec{A}} (\ell \hodgeStar_g B) - \partial_t^{A_0} (\ell^{-1} \hodgeStar_{g} E) &= - \ell \, \varphi \diamond_\Sigma (\hodgeStar_g \dif_{\vec{A}} \varphi),
	\label{eq:yangMillsHiggs3d:evol:ampere}
	\\
	\dif_{\vec{A}} (\ell^{-1} \hodgeStar_{g} E) &= \ell^{-1} \, \varphi \diamond_\Sigma (\hodgeStar_g \partial_t^{A_0} \varphi),
	\label{eq:yangMillsHiggs3dConstraint}
	\\
	- \partial_t^{A_0} (\ell^{-1} \hodgeStar_g \partial_t^{A_0} \varphi) + \dif_{\vec{A}} (\ell \hodgeStar_g \dif_{\vec{A}} \varphi) &= \ell \, V' (\varphi)\vol_g .
	\label{eq:yangMillsHiggs3d:evol:higgs}
\end{align+}\qedhere\end{subequations}
Note that~\eqref{eq:yangMillsHiggs3d:evol:ampere} is the non-abelian counterpart of Ampere's law and that~\eqref{eq:yangMillsHiggs3dConstraint} is the Gauß constraint.
Also note that \( A_0 \) is not a dynamical variable.

%%%%%%%%%%%%%%%%%%%%%%%%%%%%%%%%%%%%%%%%%%%%%%%%%%%%%%%%%%%%%%%%%%%%%%%%%%

\subsection{Formulation as a Clebsch--Lagrange system}
\label{sec:yangMillls:asClebschLagrange}

%%%%%%%%%%%%%%%%%%%%%%%%%%%%%%%%%%%%%%%%%%%%%%%%%%%%%%%%%%%%%%%%%%%%%%%%%

We will now show how the Yang--Mills--Higgs equations can be derived from the Clebsch--Lagrange variational principle.
See forthcoming \cref{table:gaugeTheory:comparisionWithGeneralTheory} for a comparison of the general Clebsch--Lagrange theory and its concrete implementation in the Yang--Mills--Higgs case.
After the \( (1 + 3) \)-splitting, the configuration space of the theory is
\begin{equation}
	\SectionSpaceAbb{Q} = \set*{(\vec{A}, \varphi) \in \ConnSpace(\vec{P}) \times \sSectionSpace(\vec{F})}.
\end{equation}
Since \( \SectionSpaceAbb{Q} \) is an affine space, its tangent bundle is trivial with fiber \( \DiffFormSpace^1(\Sigma, \AdBundle \vec{P}) \times \sSectionSpace(\vec{F}) \).
We will denote points in \( \TBundle \SectionSpaceAbb{Q} \) by tuples \( (\vec{A}, \alpha, \varphi, \zeta) \) with \( \alpha \in \DiffFormSpace^1(\Sigma, \AdBundle \vec{P}) \) and \( \zeta \in \sSectionSpace(\vec{F}) \).

A natural choice for the cotangent bundle $\CotBundle \SectionSpaceAbb{Q}$ is the trivial bundle over $\SectionSpaceAbb{Q}$ with fiber 
\begin{equation}
	\DiffFormSpace^2(\Sigma, \CoAdBundle \vec{P}) \times \DiffFormSpace^3(\Sigma, \vec{F}^*)\, .
\end{equation}
We denote elements of this fiber by pairs \( (D, \pi) \).
The natural pairing with \( \TBundle \SectionSpaceAbb{Q} \) is given by integration over \( \Sigma \), 
\begin{equation}
	\dualPair{(D, \pi)}{(\alpha, \zeta)} = \int_\Sigma (\wedgeDual{D}{\alpha} + \wedgeDual{\pi}{\zeta}) \, .
\end{equation}
As in the finite-dimensional case, the cotangent bundle \( \CotBundle \SectionSpaceAbb{Q} \) carries a natural symplectic structure.
Its canonical $1$-form $\theta$ is given by~\eqref{CanForm}.
In terms of the global coordinates \( (\vec{A}, \varphi, D, \pi) \) on \( \CotBundle \SectionSpaceAbb{Q} \) it reads\footnote{
Equivalently, one can view \( \vec{A} \) as a global coordinate function \( \CotBundle \SectionSpaceAbb{Q} \to \DiffFormSpace^1(\Sigma, \AdBundle \vec{P}) \) and \( \diF \) as the exterior differential on $ \CotBundle \SectionSpaceAbb{Q} $.
In this language, \( \diF \vec{A} \in \DiffFormSpace^1(\CotBundle \SectionSpaceAbb{Q}, \DiffFormSpace^1(\Sigma, \AdBundle \vec{P})) \) and then
\begin{equation}
	\theta = \dualPair{D}{\diF \vec{A}} + \dualPair{\pi}{\diF \varphi},
	\quad
	\Omega = \wedgeDual{\diF D}{\diF \vec{A}} + \wedgeDual{\diF \pi}{\diF \varphi}
\end{equation}
are the field theoretic counterparts of the formulae  \( \theta = p_i \dif q^i \) and \( \Omega = \dif p_i \wedge \dif q^i \) in classical mechanics.
}
\begin{equation}\label{eq:cotangentBundleTautologicalOneForm}
	\theta_{\vec{A}, \varphi, D, \pi}(\diF \vec{A}, \diF \varphi, \diF D, \diF \pi) = \dualPair{D}{\diF \vec{A}} + \dualPair{\pi}{\diF \varphi},
\end{equation}
where \( \diF \vec{A} \in \DiffFormSpace^1(\Sigma, \AdBundle \vec{P}) \), \( \diF \varphi \in \DiffFormSpace^0(\Sigma, \vec{F}) \), \( \diF D \in \DiffFormSpace^2(\Sigma, \CoAdBundle \vec{P}) \) and \( \diF \pi \in \DiffFormSpace^0(\Sigma, \vec{F}^*) \) are viewed as tangent vectors on \( \CotBundle \SectionSpaceAbb{Q} \).
The symplectic form \( \Omega \) is given as the exterior differential of \( \theta \), that is,
\begin{equation}\label{eq:yangMillsHiggs:symplecticForm}\begin{split}
	\Omega_{\vec{A}, \varphi, D, \pi}&\left((\diF \vec{A}_1, \diF \varphi_1, \diF D_1, \diF \pi_1), (\diF \vec{A}_2, \diF \varphi_2, \diF D_2, \diF \pi_2)\right) 
		\\
		&= \dualPair{\diF D_1}{\diF \vec{A}_2} - \dualPair{\diF D_2}{\diF \vec{A}_1} + \dualPair{\diF \pi_1}{\diF \varphi_2} - \dualPair{\diF \pi_2}{\diF \varphi_1}.
\end{split}\end{equation}

The group \( \GauGroup(\vec{P}) \) of local gauge transformations acts naturally on \( \SectionSpaceAbb{Q} \).
Since \( A_0 \in \sSectionSpace(\AdBundle \vec{P}) \) can be viewed as an element of \( \GauAlgebra(\vec{P}) \), it is natural to take it as the \( \xi \)-variable of the general theory.
By~\eqref{eq:generalGauge:infinitisimalGaugeAction}, the Killing vector field generated by \( A_0 \) at \( (\vec{A}, \varphi) \) is given by
\begin{equation}
	A_0 \ldot (\vec{A}, \varphi) = (- \dif_{\vec{A}} A_0, A_0 \ldot \varphi).
\end{equation}
Hence, in this case, the effective velocities \( \dot q + \xi \ldot q \) are 
\begin{equation}
	\dot{\vec{A}} - \dif_{\vec{A}} A_0 = - E
	\quad\text{and}\quad
	\dot \varphi + A_0 \ldot \varphi = \partial_t^{A_0} \varphi.
\end{equation}
The calculation~\eqref{eq:splitting:lagrangian} shows that after the \( (1 + 3) \)-decomposition the Lagrangian defined in~\eqref{eq:yangMillsHiggs:action} is of the form \( \SectionSpaceAbb{L}_\textrm{YMH} = \dif t \wedge \SectionSpaceAbb{L}_\Sigma \) with
\begin{equation}
	\SectionSpaceAbb{L}_\Sigma 
		= \frac{1}{2 \ell} \wedgeDual{E}{\hodgeStar_{g} E} - \frac{\ell}{2}\wedgeDual{B}{\hodgeStar_g B} + \frac{1}{2 \ell} \wedgeDual{\partial_t^{A_0} \varphi}{\hodgeStar_g \partial_t^{A_0} \varphi} - \frac{\ell}{2} \wedgeDual{\dif_{\vec{A}} \varphi}{\hodgeStar_g \dif_{\vec{A}} \varphi} - \ell \, V(\varphi) \vol_g .
\end{equation}
This expression shows that the Yang--Mills--Higgs action is of Clebsch--Lagrange form.
To make this precise, we define the Clebsch--Lagrangian \( \SectionSpaceAbb{L}: \TBundle \SectionSpaceAbb{Q} \times \GauAlgebra(\vec{P}) \to \R \) in the coordinates on $\TBundle \SectionSpaceAbb{Q}$ introduced above by
\begin{equation}\label{calL-YM}\begin{multlined}[c][.85\displaywidth]
	\SectionSpaceAbb{L}(\vec{A}, \alpha, \varphi, \zeta, A_0)
		= \int_\Sigma \bigg(\frac{1}{2 \ell} \wedgeDual{\alpha}{\hodgeStar_g \alpha} - \frac{\ell}{2}\wedgeDual{B}{\hodgeStar_g B} 
		\\
		+ \frac{1}{2 \ell} \wedgeDual{\zeta}{\hodgeStar_g \zeta} - \frac{\ell}{2} \wedgeDual{\dif_{\vec{A}} \varphi}{\hodgeStar_g \dif_{\vec{A}} \varphi}  - \ell \, V(\varphi) \vol_g \bigg).
\end{multlined}\end{equation}
Note that \( \SectionSpaceAbb{L} \) does not depend on the Lie algebra variable \( A_0 \).
Then, the Yang--Mills--Higgs action defined by~\eqref{eq:yangMillsHiggs:action} takes the form:
\begin{equation}\begin{split}
	\SectionSpaceAbb{S}[\vec{A}, \varphi, A_0]
		&= \int_0^T \int_\Sigma \SectionMapAbb{L}_\textrm{YMH}(A, \varphi)
		\\
		&= \int_0^T \SectionSpaceAbb{L}(\vec{A}, -E, \varphi, \partial_t^{A_0} \varphi, A_0) \dif t
		\\
		&= \int_0^T \SectionSpaceAbb{L}(\vec{A}, \dot{\vec{A}} - \dif_{\vec{A}} A_0, \varphi, \dot \varphi + A_0 \ldot \varphi, A_0) \dif t.
\end{split}\end{equation}
Comparing with~\eqref{eq:clebschLagrange:action}, we see that the variational principle associated to the Yang--Mills--Higgs action \( \SectionSpaceAbb{S} \) is a Clebsch--Lagrange principle with respect to the Clebsch--Lagrangian \( \SectionSpaceAbb{L} \).

Since the Yang--Mills--Higgs equations arise from varying the action \( \SectionSpaceAbb{S} \), the general theory in form of \cref{prop:clebschLagrange:clebschEulerLagrangeDirect} implies that the Yang--Mills--Higgs equations in its \( (1 + 3) \)-formulation~\eqref{eq:yangMillsHiggs3dEvol} are the Clebsch--Euler--Lagrange equations associated to \( \SectionSpaceAbb{L} \).
Let us verify this directly.
Since \( \SectionSpaceAbb{L} \) does not depend on \( A_0 \in \GauAlgebra(\vec{P}) \), for the remainder of this subsection we may view it as a function on \( \TBundle \SectionSpaceAbb{Q} \).
First, we have to calculate 
\begin{equation}
	\left(\difpFrac{\SectionSpaceAbb{L}}{\vec {A}}, \difpFrac{\SectionSpaceAbb{L}}{\varphi} \right) : 
	\TBundle \SectionSpaceAbb{Q} \to \R
\end{equation}
and to realize these functionals as elements of $\CotBundle \SectionSpaceAbb{Q}$.
As \( \diF F_A = \dif_A \diF A \) and \( \diF (\dif_A \beta) = \diF A \wedgeldot \beta + \dif_A \diF \beta \) for a vector-valued \( k \)-form \( \beta \), we obtain
\begin{align}
	\difpFrac{\SectionSpaceAbb{L}}{\vec{A}}
		&= - \dif_{\vec{A}} \left(\ell \, \hodgeStar_g B \right) - \ell \, \varphi \diamond_\Sigma \hodgeStar_g \dif_{\vec{A}} \varphi
		\label{eq:yangMills:difLA}
		\\
	\intertext{and}
	\difpFrac{\SectionSpaceAbb{L}}{\varphi}
		&= \dif_{\vec{A}} (\ell \hodgeStar_g \dif_{\vec{A}} \varphi) - \ell \, V'(\varphi) \vol_g \,.
		\label{eq:yangMills:difLPhi}
\end{align}
Next, the fiber derivative of \( \SectionSpaceAbb{L} \), viewed as a mapping $ \TBundle \SectionSpaceAbb{Q}  \to \CotBundle \SectionSpaceAbb{Q} $, is given by
\begin{equation}
	\left(\difpFrac{\SectionSpaceAbb{L}}{\alpha}, \difpFrac{\SectionSpaceAbb{L}}{\zeta} \right) (\vec{A}, \alpha, \varphi, \zeta)
		= \left(\ell^{-1} \hodgeStar_g \alpha, \ell^{-1} \hodgeStar_g \zeta \right).
		\label{eq:yangMills:difLFibre}
\end{equation}
We note, in particular, that \( \SectionSpaceAbb{L} \) is regular.
Evaluating~\eqref{eq:yangMills:difLFibre} at \( (\vec A, \alpha = -E, \varphi, \zeta = \partial_t^{A_0} \varphi )\), we see that the Clebsch--Euler--Lagrange equations~\eqref{eq:clebschLagrange:clebschEulerLagrangeDirect} take the following form here:
\begin{align}
	\label{CL-YM-1}
	\difFrac{}{t}(\ell^{-1} \hodgeStar_{g} E) + A_0 \ldot (\ell^{-1} \hodgeStar_{g} E) - \dif_{\vec{A}} (\ell \hodgeStar_g B) - \ell \, \varphi \diamond_\Sigma \hodgeStar_g \dif_{\vec{A}} \varphi = 0,
	\\
	\label{CL-YM-2}
	\difFrac{}{t}(\ell^{-1} \hodgeStar_g \partial_t^{A_0} \varphi) + A_0 \ldot (\ell^{-1} \hodgeStar_g \partial_t^{A_0} \varphi) - \dif_{\vec{A}} (\ell \hodgeStar_g \dif_{\vec{A}} \varphi) + \ell \, V'(\varphi) \vol_g = 0 \,.
\end{align}
These equations clearly coincide with the Yang--Mills--Higgs equations~\eqref{eq:yangMillsHiggs3d:evol:ampere} and~\eqref{eq:yangMillsHiggs3d:evol:higgs}. 

Now, let us study the momentum map constraint~\eqref{eq:clebschLagrange:hamiltonian:constraint}. 
As usual, the action on the configuration space lifts to the cotangent bundle. 
In the present setting, the \( \GauGroup(\vec{P}) \)-action on the dual variables reads
\begin{equation}
	\lambda \cdot (D, \pi) = (\CoAdAction_\lambda D, \lambda \cdot \pi).
\end{equation}
Thus, the fundamental vector field generated by \( \xi \in \GauAlgebra(\vec{P}) \) on \( \CotBundle \SectionSpaceAbb{Q} \) is given by
\begin{equation}
	\xi \ldot (\vec{A}, D, \varphi, \pi)
		= (- \dif_{\vec{A}} \xi, \CoadAction_\xi D, \xi \ldot \varphi, \xi \ldot \pi).
\end{equation}
Contracting with the canonical $1$-form \( \theta \) yields
\begin{equation}\begin{split}
	\theta (\xi \ldot (\vec{A}, D, \varphi, \pi))
		&= \int_\Sigma \bigl(
			- \wedgeDual{D}{\dif_{\vec{A}} \xi}
			+ \wedgeDual{\pi}{(\xi \ldot \varphi)} \bigr)
		\\
		&= \int_\Sigma \bigl(
			\wedgeDual{\dif_{\vec{A}} D}{\xi}
			+ \wedgeDual{(\varphi \diamond_\Sigma \pi)}{\xi} \bigr),
\end{split}\end{equation}
from which we read-off the momentum map
\begin{equation}\label{eq:yangMills:momentumMap}
	\SectionMapAbb{J}(\vec{A}, D, \varphi, \pi) = \dif_{\vec{A}} D + \varphi \diamond_\Sigma \pi,
\end{equation}
which takes values in \( \GauAlgebra(\vec{P})^* \isomorph \DiffFormSpace^3(\Sigma, \CoAdBundle \vec{P}) \).
As $\SectionSpaceAbb{L}$ does not explicitly depend on $A_0$, the momentum map constraint~\eqref{eq:clebschLagrange:hamiltonian:constraint} takes the form 
\begin{equation}
	\label{Gauss-Constr}
	\SectionMapAbb{J}(\vec{A}, D, \varphi, \pi) = 0.
\end{equation}
By~\eqref{eq:yangMills:difLFibre} the canonically conjugate momenta are
\begin{equation}\label{eq:yangMills:defConjugateMomenta}
	 (D, \pi)
	 	= \left(- \ell^{-1} \hodgeStar_g E, \ell^{-1} \hodgeStar_g \partial_t^{A_0} \varphi \right) \, .
\end{equation}
Inserting this expression into~\eqref{Gauss-Constr}, yields the Gauß constraint~\eqref{eq:yangMillsHiggs3dConstraint}.
To summarize, we have shown the following.
\begin{thm}
	The Yang--Mills--Higgs action is of Clebsch--Lagrange form.
	Moreover, the Yang--Mills--Higgs equations~\eqref{eq:yangMillsHiggs3dEvol} are equivalent to the Clebsch--Lagrange equations on \( \TBundle \SectionSpaceAbb{Q} \times \GauAlgebra(\vec{P}) \) associated to the Clebsch--Lagrangian \( \SectionMapAbb{L} \) defined in~\eqref{calL-YM}.
\end{thm}
\begin{table}[tbp]
	\scriptsize
	\centering
	\begin{tabular}{l l l}
		\toprule
			 &
			Clebsch--Lagrange &
			Yang--Mills--Higgs
			\\
		\midrule
			Configuration variables &
			\( q \) &
			\( \vec{A}, \varphi \)
			\\
			Symmetry variables &
			\( \xi \in \LieA{g} \)&
			\( A_0 \in \sSectionSpace(\AdBundle \vec{P}) \)
			\\
			Effective velocities & 
			\( \dot q + \xi \ldot q \) & 
			\( -E, \difp_t^{A_0} \varphi \)
			\\
			Conjugate momenta &
			\( p \) &
			\( D, \pi \)
			\\
			\addlinespace
			\addlinespace
			\parbox[c]{3.6cm}{Clebsch--Euler--Lagrange \\ equation} &
			\( \difFrac{}{t} \left( \difpFrac{L}{\dot{q}} \right) + \xi \ldot \difpFrac{L}{\dot{q}} = \difpFrac{L}{q} \) &
			\(
			\begin{gathered}
				\dif_{\vec{A}} (\ell \hodgeStar_g B) - \partial_t^{A_0} (\ell^{-1} \hodgeStar_{g} E) = - \ell \, \varphi \diamond_\Sigma (\hodgeStar_g \dif_{\vec{A}} \varphi),
				\\
				- \partial_t^{A_0} (\ell^{-1} \hodgeStar_g \partial_t^{A_0} \varphi) + \dif_{\vec{A}} (\ell \hodgeStar_g \dif_{\vec{A}} \varphi) = \ell \, V' (\varphi)\vol_g
			\end{gathered}
			\)
			\\
			\addlinespace
			\addlinespace
			Momentum map constraint &
			\( J\left(q, \difpFrac{L}{\dot q}\right) =  - \difpFrac{L}{\xi} \) & 
			\( \dif_{\vec{A}} D + \varphi \diamond_\Sigma \pi = 0 \)
			\\
			\addlinespace
			\addlinespace
			\parbox[c]{3.6cm}{\vspace*{-2ex}Clebsch--Hamilton \\ equations} &
			\(
			\begin{aligned}
				\difpFrac{H}{q}(q, p, \xi) &= - (\dot p + \xi \ldot p),
				\\
				\difpFrac{H}{p}(q, p, \xi) &= \dot q + \xi \ldot q
			\end{aligned}
			\)
			&
			\(
			\begin{aligned}
				\dif_{\vec{A}} H - \partial_t^{A_0} D &= - \varphi \diamond_\Sigma \psi,
				\\
				\partial_t^{A_0}\pi + \dif_{\vec{A}} \psi &= \ell \, V' (\varphi)\vol_g
			\end{aligned}
			\)
			\\
		\bottomrule
	\end{tabular}
	\caption{Comparison of the general Clebsch--Lagrange theory and the Yang--Mills--Higgs system.}
	\label{table:gaugeTheory:comparisionWithGeneralTheory}
\end{table}

%%%%%%%%%%%%%%%%%%%%%%%%%%%%%%%%%%%%%%%%%%%%%%%%%%%%%%%%%%%%%%%%%%%%%%%%%%%%%%%%%%%%%%%%%

\subsection{Hamiltonian picture}

%%%%%%%%%%%%%%%%%%%%%%%%%%%%%%%%%%%%%%%%%%%%%%%%%%%%%%%%%%%%%%%%%%%%%%%%%%%%%%%%%%%%%%%%%%%

It is straightforward to spell out the results of \cref{Ham-Pic-gen} for the model under consideration.
Therefore, we limit ourselves to the main points. 
By definition, the Clebsch--Hamiltonian \( \SectionSpaceAbb{H}: \CotBundle \SectionSpaceAbb{Q} \times \GauAlgebra(\vec{P}) \to \R \) is the Legendre transform of \( \SectionSpaceAbb{L} \). 
That is, 
\begin{equation}
\label{calH}
	\SectionSpaceAbb{H}(\vec{A}, D, \varphi, \pi, A_0)
		= \dualPair{D}{\alpha} + \dualPair{\pi}{\zeta} - \SectionSpaceAbb{L}(\vec{A}, \alpha, \varphi, \zeta, A_0),
\end{equation}
where \( \alpha \in \sSectionSpace(\AdBundle(\vec{P})) \) and \( \zeta \in \sSectionSpace(\vec{F}) \) are considered as  functions of \( (\vec{A}, D, \varphi, \pi) \) via the condition (\cf,~\eqref{eq:yangMills:difLFibre})
\begin{equation}
	(D, \pi) 
		= \left(\difpFrac{\SectionSpaceAbb{L}}{\alpha}, \difpFrac{\SectionSpaceAbb{L}}{\zeta} \right) (\vec{A}, \alpha, \varphi, \zeta)
		= \left(\ell^{-1} \hodgeStar_g \alpha, \ell^{-1} \hodgeStar_g \zeta \right).
\end{equation}
Since \( \SectionMapAbb{L} \) does not depend on \( A_0 \), the Clebsch--Hamiltonian \( \SectionMapAbb{H} \) is independent of \( A_0 \), too.
By~\eqref{calL-YM}, we have
\begin{equation}\label{eq:yangMills:hamiltonian}\begin{split}
	\SectionSpaceAbb{H}(\vec{A}, D, \varphi, \pi, A_0)
		&= \dualPair{D}{\ell \hodgeStar_g D} + \dualPair{\pi}{\ell \hodgeStar_g \pi} - \SectionSpaceAbb{L}(\vec{A}, \ell \hodgeStar_g D, \varphi, \ell \hodgeStar_g \pi, A_0),
		\\
		& 
		\!\begin{multlined}[b][.55\displaywidth]
     		= \int_\Sigma \bigg( \ell \wedgeDual{D}{\hodgeStar_g D}  + \ell \wedgeDual{\pi}{\hodgeStar_g \pi} - \frac{\ell}{2} \wedgeDual{D}{\hodgeStar_g D} + \frac{\ell}{2}\wedgeDual{B}{\hodgeStar_g B} 
		\\
		- \frac{\ell}{2} \wedgeDual{\pi}{\hodgeStar_g \pi} + \frac{\ell}{2} \wedgeDual{\dif_{\vec{A}} \varphi}{\hodgeStar_g \dif_{\vec{A}} \varphi} + \ell \, V(\varphi) \vol_g \bigg)
     		\end{multlined}
		\\
		&
		\!\begin{multlined}[b][.71\displaywidth]
		= \int_\Sigma \frac{\ell}{2} \bigg( \wedgeDual{D}{\hodgeStar_g D} +  \wedgeDual{B}{\hodgeStar_g B}  
		\\
		+ \wedgeDual{\pi}{\hodgeStar_g \pi} + \wedgeDual{\dif_{\vec{A}} \varphi}{\hodgeStar_g \dif_{\vec{A}} \varphi} + 2 \, V(\varphi) \vol_g \bigg) \, .
		\end{multlined}
\end{split}\end{equation}
In summary, the general theory from \cref{Ham-Pic-gen} yields the following.
\begin{thm}
	The Clebsch--Hamilton equations for \( \SectionMapAbb{H} \) are given by the following system of equations on \( \CotBundle \SectionSpaceAbb{Q} \times \GauAlgebra(\vec{P}) \):
	\begin{equation}\label{eq:yangMills:asClebschHamilton}\begin{split}
		\partial_t^{A_0} D &= - \dif_{\vec{A}} (\ell \hodgeStar_g B) - \ell \varphi \diamond \hodgeStar_g \dif_{\vec{A}} \varphi,
		\\
		\partial_t \vec{A} &= \dif_{\vec{A}} A_0 + \ell \hodgeStar_g D,
		\\
		\partial_t^{A_0}\pi &= \dif_{\vec{A}} (\ell \hodgeStar_g \dif_{\vec{A}} \varphi) - \ell \, V'(\varphi) \vol_g,
		\\
		\partial_t^{A_0}\varphi &= \ell \hodgeStar_g \pi,
		\\
		\dif_{\vec{A}} D &+ \varphi \diamond_\Sigma \pi = 0 \, .
	\end{split}\end{equation}
	These equations are equivalent to the Yang--Mills--Higgs equations~\eqref{eq:yangMillsHiggs3dEvol}.
\end{thm}
\begin{proof}
	Since the Lagrangian \( \SectionMapAbb{L} \) is regular, the claim follows directly from the general equivalence of the Clebsch--Lagrange and Clebsch--Hamilton equations under the Legendre transformation, see \cref{prop:clebschLagrange:legendre}.
	For completeness, let us also give a direct proof, that is, let us write down the Clebsch--Hamilton equations~\eqref{eq:clebschLagrange:hamiltonian} for the case under consideration. 
	From~\eqref{eq:yangMills:hamiltonian}, we obtain
	\begin{equation}
		\difpFrac{\SectionMapAbb{H}}{\vec{A}} = \dif_{\vec{A}} (\ell \hodgeStar_g B) + l \varphi \diamond \hodgeStar_g \dif_{\vec{A}} \varphi
	\end{equation}
	and
	\begin{equation}
		\difpFrac{\SectionMapAbb{H}}{\varphi} = - \dif_{\vec{A}} (\ell \hodgeStar_g \dif_{\vec{A}} \varphi) + \ell \, V'(\varphi) \vol_g.
	\end{equation}
	Moreover, we obviously have
	\begin{equation}
		\difpFrac{\SectionMapAbb{H}}{D} = \ell \hodgeStar_g D,
		\qquad
		\difpFrac{\SectionMapAbb{H}}{\pi} = \ell \hodgeStar_g \pi,
		\qquad
		\difpFrac{\SectionMapAbb{H}}{A_0} = 0.
	\end{equation}
	As we have already seen, the constraint~\eqref{constr-H-ext} reads
	\begin{equation}
		\dif_{\vec{A}} D + \varphi \diamond_\Sigma \pi =0
	\end{equation}
	and corresponds to the Gauß constraint~\eqref{Gauss-Constr}.
	Thus the Clebsch--Hamilton equations~\eqref{eq:clebschLagrange:hamiltonian} yield~\eqref{eq:yangMills:asClebschHamilton}.
	Using expression~\eqref{eq:yangMills:defConjugateMomenta} for the canonical momenta, it is straightforward to see that~\eqref{eq:yangMills:asClebschHamilton} recovers the Yang--Mills--Higgs equations~\eqref{eq:yangMillsHiggs3dEvol}, indeed.
\end{proof}

\begin{remark}
	Note that the metric \( g \) on \( \Sigma \) depends on time.
	Accordingly, the Hamiltonian is explicitly time-dependent and so it is not a constant of motion for curved spacetimes.
\end{remark}

%%%%%%%%%%%%%%%%%%%%%%%%%%%%%%%%%%%%%%%%%%%%%%%%%%%%%%%%%%%%%%%%%%%%%%%%%%%%%%%%

\subsection{Constraints and reduction by stages}

%%%%%%%%%%%%%%%%%%%%%%%%%%%%%%%%%%%%%%%%%%%%%%%%%%%%%%%%%%%%%%%%%%%%%%%%%%%%%%

Recall that the Clebsch--Hamilton equations can also be obtained by passing through an extended phase space and implementing the constraints using the Dirac--Bergmann algorithm in the formulation of \parencite{GotayNesterHinds1978}.
Here, we apply the general theory of \cref{Ham-Pic-gen} to the Yang--Mills--Higgs case.

By definition, the extended configuration space is \( \SectionSpaceAbb{Q}_\ext \defeq \SectionSpaceAbb{Q} \times \GauAlgebra(\vec{P}) \).
Thus, the  natural extended phase space induced by Hodge duality is
\begin{equation}
 	\CotBundle \SectionSpaceAbb{Q}_\ext \defeq \CotBundle \SectionSpaceAbb{Q} \times \GauAlgebra(\vec{P}) \times \DiffFormSpace^3(\Sigma, \CoAdBundle \vec{P}).
\end{equation}
We will denote points in \( \CotBundle \SectionSpaceAbb{Q}_\ext \) by tuples \( (\vec{A}, D, \varphi, \pi, A_0, \nu) \).
As in the general theory, associated with the Clebsch--Lagrangian \( \SectionSpaceAbb{L}: \TBundle \SectionSpaceAbb{Q} \times \GauAlgebra(\vec{P}) \to \R \) as in~\eqref{calL-YM}, we have a degenerate Lagrangian $\SectionSpaceAbb{L}_\ext$ given by~\eqref{L-ext} and we can, thus, perform the Dirac-Bergmann constraint analysis as in \cref{Ham-Pic-gen}. 
As in the general theory, the primary constraint is \( \nu = 0 \).
Let \( \SectionSpaceAbb{M}_1 \isomorph \CotBundle \SectionSpaceAbb{Q} \times \GauAlgebra(\vec{P}) \) be the subspace cut out by this constraint.
Let $ \SectionSpaceAbb{H}_\ext: \SectionSpaceAbb{M}_1 \to \R $ be the Hamiltonian corresponding to $\SectionSpaceAbb{L}_\ext$.
Using~\eqref{eq:compositeHamiltonian} and~\eqref{eq:yangMills:momentumMap}, from~\eqref{H-ext} we read off 
\begin{equation}
	\label{calHext}
	\SectionSpaceAbb{H}_\ext (\vec{A}, D, \varphi, \pi, A_0) = \SectionSpaceAbb{H}(\vec{A}, D, \varphi, \pi, A_0) - \dualPair{\dif_{\vec{A}} D + \varphi \diamond_\Sigma \pi}{A_0}.
\end{equation}
The next step in the constraint analysis leads to the constraint~\eqref{constr-H-ext} and, thus, to the system of Hamilton equations given by~\eqref{eq:clebschLagrange:hamiltonianExt}.
We have already seen that the constraint~\eqref{constr-H-ext} coincides with the Gauß constraint.
As $\SectionSpaceAbb{H}$ is gauge invariant and does not explicitly depend on $A_0$, \cref{prop:G-inv-H} implies that the Gauß constraint is 
preserved in time.
Consequently, the constraint analysis terminates at that stage. 
The Hamiltonian \( \SectionMapAbb{H}_\ext \) transforms under local gauge transformations as follows:
\begin{equation}\begin{split}
	\SectionMapAbb{H}_\ext&(\AdAction_\lambda \vec{A} + \lambda \dif \lambda^{-1}, \CoAdAction_\lambda D, \lambda \cdot \varphi, \lambda \cdot \pi, \AdAction_\lambda A_0 - \xi)
		\\
		&= \SectionMapAbb{H}_\ext(\vec{A}, D, \varphi, \pi, A_0)
		+ \int_\Sigma \dualPair{(\dif_{\vec{A}} D + \varphi \diamond_\Sigma \pi)}{\AdAction_{\lambda}^{-1} \xi}.
\end{split}\end{equation}
In particular, \( \SectionMapAbb{H}_\ext \) is not gauge invariant unless the Gauß constraint is imposed.
Using the terminology common in physics, one could say that the extended Hamiltonian is invariant \enquote{on-shell}.

According to the discussion of the general reduction theory in \cref{sec:clebschLagrange:reductionStages}, the action of \( \TBundle \GauGroup(\vec{P}) \) on \( \CotBundle \SectionSpaceAbb{Q}_\ext \) plays an important role.
We will now show that, in the present context, this action is derived from the action of \( \GauGroup \) on \( \ConnSpace \times \SectionSpaceAbb{F} \).
Recall from~\eqref{LocGTr} the action of the group of local gauge transformations \( \GauGroup \) on the space of connections \( \ConnSpace \),
\begin{equation}
	\lambda \cdot A = \AdAction_\lambda A - \difLog^R \lambda,
\end{equation}
where \( \lambda \) is a \( G \)-equivariant map \( P \to G \) and \( \difLog^R \lambda \defeq \dif \lambda \, \lambda^{-1} \in \DiffFormSpace^1(P, \LieA{g}) \) denotes the right logarithmic derivative.
Via the $(1+3)$-decomposition, we may equivalently think of a gauge transformation \( \lambda \in \GauGroup \) as a time-dependent gauge transformation \( \lambda(t) \) in \( \vec{P} \).
Evaluating $\lambda \cdot A $ on a vector field $Y$ on $P$, decomposed as \( Y = Y^0 \difp_t + \vec{Y} \), we obtain
\begin{equation}\begin{split}
	(\lambda \cdot A) (Y^0 \difp_t + \vec{Y})
		&= Y^0 \AdAction_\lambda A_0 - Y^0 \difLog^R_t \lambda + \AdAction_\lambda \vec{A}(\vec{Y}) - \difLog^R \lambda (\vec{Y})
		\\
		&= Y^0 (\AdAction_\lambda A_0 - \difLog^R_t \lambda) + (\lambda \cdot \vec{A})(\vec{Y}),
\end{split}\end{equation}
where \( \difLog^R_t \lambda \in \sSectionSpace(\AdBundle \vec{P}) \) denotes the right logarithmic derivative of the path \( t \mapsto \lambda(t) \in \GauGroup(\vec{P}) \).
That is,
\begin{equation}
	\difLog^R_t \lambda
		= \difLog^R \lambda (\difp_t)
		= \tangent \lambda (\difp_t) \ldot \lambda^{-1}
		\equiv \dot{\lambda} \, \lambda^{-1}
		\in \GauAlgebra(\vec{P}).
\end{equation}
Thus, the action of gauge transformations on connections on \( P \) decomposes into the action \( A_0 \mapsto \AdAction_\lambda A_0 - \difLog^R_t \lambda \) and the natural action on the space  $\ConnSpace (\vec P)$ of connections on \( \vec{P} \).
The crucial observation is that the action only depends on \( \lambda(t) \in \GauGroup(\vec{P}) \) and on the first derivative \( \difLog^R_t \lambda \).
In other words, the action of the group of time-dependent gauge transformations on the pair \( (A_0, \vec{A}) \) factors through the following action of the tangent group \( \TBundle \GauGroup(\vec{P}) \):
\begin{equation}\begin{split}
	(\xi, \lambda) \cdot A_0 &= \AdAction_\lambda A_0 - \xi,
	\\
	(\xi, \lambda) \cdot \vec{A} &= \AdAction_\lambda \vec{A} - \difLog^R \lambda,
\end{split}\end{equation}
where \( (\xi, \lambda) \in \GauAlgebra(\vec{P}) \times \GauGroup(\vec{P}) \) is viewed as the element of \( \TBundle \GauGroup(\vec{P}) \) under the right trivialization \( (\xi, \lambda) \mapsto \xi \ldot \lambda \).
The action of \( \GauGroup \) on the matter field \( \varphi \in \sSectionSpace(F) \) is considerably simpler, because it does not involve derivatives in the time direction and thus comes down to an action of \( \GauGroup(\vec{P}) \) on \( \sSectionSpace(\vec{F}) \) for every moment of time.
In summary, we get a natural action of \( \TBundle \GauGroup(\vec{P}) \) on the extended configuration space of the theory.

Moreover, \( \TBundle \GauGroup(\vec{P}) = \GauAlgebra(\vec{P}) \rSemiProduct \GauGroup(\vec{P}) \).
As \( \SectionSpaceAbb{H} \) is \( \GauGroup(\vec{P}) \)-invariant and does not explicitly depend on \( A_0 \), \cref{prop:clebschLagrange:reductionStages} holds and we have the following commutative diagram of symplectic reductions:
\begin{equation}\begin{tikzcd}[column sep=5.9em, row sep=0.02em]
	\CotBundle \SectionSpaceAbb{Q}_\ext
		\ar[r, twoheadrightarrow, "{\sslash \GauAlgebra(\vec{P})}"]
		\ar[rr, twoheadrightarrow, "{\sslash \TBundle \GauGroup(\vec{P})}", swap, bend right]
	& \CotBundle \SectionSpaceAbb{Q}
		\ar[r,twoheadrightarrow, "{\sslash \GauGroup(\vec{P})}"]
	& \check{\SectionSpaceAbb{M}},
\end{tikzcd}\end{equation}
where \( \check{\SectionSpaceAbb{M}} = \SectionMapAbb{J}^{-1}(0) \slash \GauGroup(\vec{P}) \) is the reduced phase space of the theory.
In the first reduction step, we pass from the variables \( (\vec{A}, D, \varphi, \pi, A_0, \nu) \) to the variables \( (\vec{A}, D, \varphi, \pi) \).
In this language, the temporal gauge \( A_0 = 0 \) often used in physics acquires the geometric interpretation of a section in \( \CotBundle \SectionSpaceAbb{Q} \times \GauAlgebra(\vec{P}) \to \CotBundle \SectionSpaceAbb{Q} \).
Moreover, \cref{prop:clebschLagrange:equivalenceHamiltonianReductions} immediately gives the following.
\begin{thm}
	The following systems of equations are equivalent:
	\begin{thmenumerate}
		\item
			The Yang--Mills--Higgs equations~\eqref{eq:yangMillsHiggs4d}.
		\item
			The Clebsch--Hamilton equations on \( \CotBundle \SectionSpaceAbb{Q} \times \GauAlgebra(\vec{P}) \) with respect to \( \SectionMapAbb{H} \).
		\item
			The constraint Hamilton equations on \( \CotBundle \SectionSpaceAbb{Q}_\ext \) with respect to the Hamiltonian \( \SectionMapAbb{H}_\ext \).
		\item
			The Hamilton equations on \( \check{\SectionSpaceAbb{M}} \) with respect to the reduced Hamiltonian \( \check{\SectionMapAbb{H}}: \check{\SectionSpaceAbb{M}} \to \R \) defined by
			\begin{equation}
				\pi^* \check{\SectionMapAbb{H}} = \restr{\SectionMapAbb{H}}{\SectionMapAbb{J}^{-1}(0)},
			\end{equation}
			where \( \pi: \SectionMapAbb{J}^{-1}(0) \to \check{\SectionSpaceAbb{M}} \) is the natural projection.
			\qedhere 
	\end{thmenumerate}
\end{thm}
Since the Clebsch--Hamiltonian \( \SectionMapAbb{H}: \CotBundle \SectionSpaceAbb{Q} \times \GauAlgebra(\vec{P}) \to \R \) is \( A_0 \)-independent, it may be viewed as an ordinary Hamiltonian on \( \CotBundle \SectionSpaceAbb{Q} \).
We emphasize that the ordinary Hamilton equations on \( \CotBundle \SectionSpaceAbb{Q} \) with respect to \( \SectionMapAbb{H} \) are \emph{not} equivalent to the Yang--Mills--Higgs equations unless the Gauß constraint is also imposed.
The geometry of the reduced phase space \( \check{\SectionSpaceAbb{M}} \) will be studied in \parencite{DiezRudolphReduction}.

\section{General relativity}
We will now consider Einstein's equation of general relativity and show how it fits into the general Clebsch--Lagrange variational framework.

Let \( M \) be a \( 4 \)-dimensional oriented manifold.
We are interested in solutions of Einstein's equation
\begin{equation}
	\Ric_\eta - \frac{1}{2} \RicScalar_\eta \, \eta = 8 \pi T,
\end{equation}
where \( \Ric_\eta \) and \( \RicScalar_\eta \) denote, respectively, the Ricci curvature and the Ricci scalar curvature of the Lorentzian metric \( \eta \) on \( M \) to be determined.
For simplicity, we will restrict attention to the vacuum setting, for which the energy-momentum tensor \( T \) vanishes and Einstein's equation is equivalent to \( \Ric_\eta = 0 \).
Equivalently, we are looking for extrema of the Einstein--Hilbert action
\begin{equation}\label{eq:generalRelativity:einsteinHilbert4D}
	\SectionSpaceAbb{S}[\eta] = \frac{1}{2} \int_M \RicScalar_\eta \, \vol_\eta,
\end{equation}
defined on the space of Lorentzian metrics on \( M \) with signature \( (- + +\, +) \).

In order to formulate~\eqref{eq:generalRelativity:einsteinHilbert4D} as a Clebsch--Lagrange variational problem, we proceed similarly to the study of the Yang--Mills--Higgs equations:
\begin{enumerate}
	\item 
		Using a splitting of spacetime, formulate Einstein's equation as a Cauchy problem in the variables \( (S, \ell, g) \), where \( S \) is the shift vector, \( \ell \) is the lapse function and \( g \) is the spatial metric.
	\item
		Realize \( S \) and \( \ell \) as elements of the Lie algebra of the diffeomorphism group of \( M \) and determine their action on the space of spatial metrics.
	\item
		Calculate the Lagrangian in the \( (1+3) \)-splitting and show that it is of the Clebsch--Lagrange form.
	\item
		Determine the momentum map constraint and compare it to the diffeomorphism and Hamiltonian constraint.
	\item
		Pass to the Hamiltonian picture using the Clebsch--Legendre transformation.
\end{enumerate}

First, we formulate Einstein's equation as a Cauchy problem, see \eg \parencite[Section~10.2]{Wald1984}.
For this purpose, choose a slicing of \( M \), that is, a diffeomorphism \( \iota: \R \times \Sigma \to M \), where \( \Sigma \) is a compact manifold that will play the role of a Cauchy hypersurface.
For simplicity, we assume that \( \Sigma \) has no boundary\footnote{We refer to \parencite[Chapter~20]{Blau2018} for a detailed discussion of boundary terms within the ADM formalism and to \parencite{Kijowski1997} for a new canonical description of the gravitational field dynamics in a finite volume with boundary. Clearly, our approach can be adapted to include boundary terms as well.}.
For every \( t \in \R \), let \( \iota_t: \Sigma \to M \) denote the smooth curve of embeddings associated to \( \iota \).
Assume that these embeddings are spacelike.
Let \( \Sigma_t = \iota_t (\Sigma) \) be the corresponding submanifold of \( M \) and denote its unit\footnote{\( \eta(\nu, \nu) = -1 \).} normal vector field by \( \nu \).
Since \( \iota \) is a diffeomorphism, we have the following splitting of the tangent space 
\begin{equation}
	\TBundle_{\iota(t, x)} M = \R \, \nu_{\iota(t, x)} \oplus \tangent_{(t,x)} \iota (\TBundle_x \Sigma)
\end{equation}
for every \( t \in \R \) and \( x \in \Sigma \).
Accordingly, every vector field on \( M \) decomposes into parts normal and tangent to \( \Sigma_t \).
In particular, for the time vector field, we obtain
\begin{equation}
	\difpFracAt{}{t}{t} \iota_t = \tangent \iota (\difp_t) = \ell \, \nu + \tangent \iota_t (S),
\end{equation} 
where \( \ell(t) \in \sFunctionSpace(\Sigma) \) is the \emphDef{lapse function} and \( S(t) \in \VectorFieldSpace(\Sigma) \) is the \emphDef{shift vector field}.
Moreover, \( \eta \) induces by pull-back a family \( g(t) = \iota^*_t \eta \) of Riemannian metrics on \( \Sigma \).
In terms of these data, the pull-back of \( \eta \) to \( \R \times \Sigma \) is given by 
\begin{equation}
	\iota^* \eta = \Matrix{- \ell^2 + g (S, S) & g(S, \cdot) \\ g(S, \cdot) & g}.
\end{equation}
In particular, the knowledge of the data \( (S, \ell, g) \) is enough to reconstruct \( \eta \) (in a neighborhood around \( \Sigma_0 \)).

\subsection{Formulation as a Clebsch--Lagrange system}
In order to make contact to the general Clebsch--Lagrange theory, we need to identify the Lie algebra-valued fields and the configuration space on which the Lie group acts.
Since the lapse function and the shift vector field are known to be related to the diffeomorphism invariance of the theory, they are natural candidates for the Lie algebra-valued fields leaving the metric \( g \) as the configuration variable.
To make this idea precise, let \( \SectionSpaceAbb{G} \) be the Lie group consisting of diffeomorphisms \( \phi \) of \( \R \times \Sigma \) that are of the form \( \phi(t, x) = (t + \bar{\phi}(x), \varphi(x)) \), where \( \bar{\phi} \) is a real-valued function on \( \Sigma \) and \( \varphi \) is a diffeomorphism of \( \Sigma \).
Using the slicing \( \iota \) we can view \( \SectionSpaceAbb{G} \) as a subgroup of the group \( \DiffGroup(M) \) of diffeomorphisms of \( M \).
If \( \R \times \Sigma \to \Sigma \) is viewed as a principal \( \R \)-bundle, then \( \SectionSpaceAbb{G} \) is identified with the group of principal bundle automorphism.
Clearly, \( \SectionSpaceAbb{G} \) has the following semidirect product structure:
\begin{equation}
	\SectionSpaceAbb{G} = \DiffGroup(\Sigma) \lSemiProduct \sFunctionSpace(\Sigma),
\end{equation}
where the group \( \DiffGroup(\Sigma) \) of diffeomorphisms of \( \Sigma \) acts on \( \sFunctionSpace(\Sigma) \) by pull-back.
Thus, the Lie algebra \( \mathrm{L}\SectionSpaceAbb{G} \) of \( \SectionSpaceAbb{G} \) is the semidirect product of \( \VectorFieldSpace(\Sigma) \) and \( \sFunctionSpace(\Sigma) \).
Hence, \( (S, \ell) \) are naturally elements of \( \mathrm{L}\SectionSpaceAbb{G} \).
Correspondingly, the configuration space \( \SectionSpaceAbb{Q} \) of the theory is the space \( \MetricSpace(\Sigma) \) of Riemannian metrics on \( \Sigma \).
There is a natural left action of \( \SectionSpaceAbb{G} \) on \( \MetricSpace(\Sigma) \) through the projection onto \( \DiffGroup(\Sigma) \):
\begin{equation}
	(\varphi, \bar{\phi}) \cdot g \defeq (\varphi^{-1})^* g .
\end{equation}
Correspondingly, \( (X, f) \in \mathrm{L}\SectionSpaceAbb{G} \) acts on \( g \in \MetricSpace(\Sigma) \) as 
\begin{equation}
	\label{eq:generalRelativity:actionLieAlgebra}
	(X, f) \ldot g = - \difLie_X g.
\end{equation}
Recall (\eg, from \parencite[Equation~232]{Giulini2015}) that the variation of \( g \) is connected to the extrinsic curvature \( k(t) \in \SymTensorFieldSpace^2(\Sigma) \) by
\begin{equation}
	\dot{g} \equiv \difFracAt{}{t}{t} g(t) = \difLie_S g + 2 \ell k,
\end{equation}
where \( \difLie_S \) is the Lie derivative along \( S \).
Thus, by~\eqref{eq:generalRelativity:actionLieAlgebra}, the effective velocity \( \dot{q} + \xi \ldot q \) of the general Clebsch--Lagrange theory is given by
\begin{equation}
	\label{eq:generalRelativity:effectiveVelocity}
	\dot{g} + (S, \ell) \ldot g = \difLie_S g + 2 \ell k - \difLie_S g = 2 \ell k.
\end{equation}

Let us recall how the Einstein--Hilbert action looks like in terms of the variables \( (S, \ell, g) \), see, \eg, \parencites[Appendix~E.2]{Wald1984}[Section~5 and~6]{Giulini2015}.
Using Gauß' formula, the Ricci scalar curvature \( \RicScalar_\eta \) of \( \eta \) can be written in terms of the Ricci scalar curvature \( \RicScalar_g \) of \( g(t) \) and the second fundamental form \( k(t) \in \SymTensorFieldSpace^2(\Sigma) \) of the embedding \( \iota_t \):
\begin{equation}
	\label{eq:generalRelativity:ricciScalar}
	\RicScalar_\eta = \RicScalar_g - 2 \Ric_\eta (\nu, \nu) - \norm{k}^2_g + (\tr_g k)^2.
\end{equation}
Moreover, we have the Ricci equation,
\begin{equation}
	\label{eq:generalRelativity:ricciEq}
	\Ric_\eta(\nu, \nu) = (\tr_g k)^2 - \norm{k}^2_g + \divergence_g v,
\end{equation}
where \( v(t) \in \VectorFieldSpace(\Sigma) \) is a certain time-dependent vector field.
We assume that \( M \) is oriented and time-oriented in such a way that \( \iota^* \vol_\eta = \ell \dif t \wedge \vol_g \).
%\footnotetext{By the Leibniz formula for determinants, \( \det (\iota^* \eta) = - \ell_t^2 \det g \).}%
Then, inserting~\eqref{eq:generalRelativity:ricciEq} into~\eqref{eq:generalRelativity:ricciScalar} and using the fact \( \Sigma \) has no boundary, the Einstein--Hilbert action takes the following form relative to the splitting \( \iota \):
\begin{equation}\label{eq:generalRelativity:einsteinHilbert3D}
	\SectionSpaceAbb{S}[g, S, \ell]
		= \int_0^T \dif t \, \int_\Sigma \ell \left(\RicScalar_{g} + \norm{k}^2_{g} - (\tr_{g} k)^2\right) \, \vol_{g}.
\end{equation}
The space of metrics on \( \Sigma \) is an open cone in the space \( \SymTensorFieldSpace^2(\Sigma) \) of symmetric \( 2 \)-tensors.
Thus, the tangent bundle \( \TBundle \MetricSpace(\Sigma) \) may be identified with \( \MetricSpace(\Sigma) \times \SymTensorFieldSpace^2(\Sigma) \).
We will denote elements of \( \TBundle \MetricSpace(\Sigma) \) by pairs \( (g, h) \) with \( g \in \MetricSpace(\Sigma) \) and \( h \in \SymTensorFieldSpace^2(\Sigma) \).
In terms of these variables, define the Clebsch--Lagrangian \( \SectionSpaceAbb{L}: \TBundle \MetricSpace(\Sigma) \times \mathrm{L} \SectionSpaceAbb{G} \to \R \) by
\begin{equation}
	\SectionSpaceAbb{L}(g, h, S, \ell) = \int_\Sigma \left( \ell \, \RicScalar_{g} + \frac{1}{4 \ell} \norm{h}^2_{g} - \frac{1}{4 \ell} (\tr_{g} h)^2\right) \, \vol_{g},
\end{equation}
which is the counterpart of \( L(q, \dot q, \xi) \) in the general theory.
Then, using~\eqref{eq:generalRelativity:effectiveVelocity}, the Einstein--Hilbert action~\eqref{eq:generalRelativity:einsteinHilbert3D} reads as follows:
\begin{equation}
 	\label{eq:generalRelativity:einsteinHilbert3DClebschForm}
	\SectionSpaceAbb{S}[g, S, \ell]
		= \int_0^T \dif t \, \SectionSpaceAbb{L}(g, 2 \ell k, S, \ell) 
		= \int_0^T \dif t \, \SectionSpaceAbb{L}(g, \dot{g} + (S, \ell) \ldot g, S, \ell).
\end{equation} 
Hence, we see that the Einstein--Hilbert action is of the general Clebsch--Lagrange form~\eqref{eq:clebschLagrange:action} with Clebsch--Lagrangian \( \SectionSpaceAbb{L} \).
We emphasize that, in sharp contrast to the Yang--Mills case, the Clebsch--Lagrangian of general relativity depends on the Lie algebra variable \( \xi = (S, \ell) \), more precisely, it depends on \( \ell \) but not on \( S \).

Now, consider the cotangent bundle of \( \MetricSpace(\Sigma) \).
The dual space to \( \SymTensorFieldSpace^2(\Sigma) \) with respect to the integration pairing is the space \( \SymTensorFieldSpace_2(\Sigma; \VolSpace) \) of volume-form-valued contravariant symmetric \( 2 \)-tensors.
Hence, it is natural to define \( \CotBundle \MetricSpace(\Sigma) = \MetricSpace(\Sigma) \times \SymTensorFieldSpace_2(\Sigma; \VolSpace) \) endowed with the natural symplectic structure.
We will denote the conjugate momenta by \( \pi \).
By \cref{prop:clebschLagrange:legendre}, the associated constraint is given by the momentum map \( \SectionSpaceAbb{J} \) for the lifted \( \DiffGroup(\Sigma) \)-action on \( \CotBundle \MetricSpace(\Sigma) \).
As in the general theory, denote the natural integration paring of \( \mathrm{L}\SectionSpaceAbb{G} \) with \( \mathrm{L}\SectionSpaceAbb{G}^* = \DiffFormSpace^1(\Sigma, \VolSpace) \times \VolSpace(\Sigma) \) by \( \kappa \).
Now, consider the momentum map given by~\eqref{eq:cotangentBundle:momentumMapDef}.
Using~\eqref{eq:generalRelativity:actionLieAlgebra} together with the Koszul identity and metric compatibility, we calculate
\begin{align}
	\kappa\left(\SectionSpaceAbb{J}(g, \pi), (X, f)\right)
		&= \dualPair{\pi}{(X, f) \ldot g}
		\\
		&= - \int_\Sigma \dualPair{\pi}{\difLie_X g}
		\\
		&= - 2 \int_\Sigma \dualPair{\pi}{\Sym g(\nabla_{(\cdot)} X, \cdot)}
		\\
		&= - 2 \int_\Sigma \dualPair{\pi}{\Sym (\nabla_{(\cdot)} X^\flat) (\cdot)},
\intertext{where \( X^\flat(Y) = g(X, Y) \). Since \( \Sigma \) has no boundary, integration by parts yields}
	\kappa\left(\SectionSpaceAbb{J}(g, \pi), (X, f)\right)
		&= 2 \int_\Sigma \dualPair{\divergence_g \pi}{X^\flat},
\end{align}
where the divergence of \( \pi \) is defined as the trace of \( \nabla_{(\cdot)} \pi (\cdot, \cdot) \), that is, in abstract index notation \( (\divergence_g \pi)^j = \nabla_i \pi^{ij} \).
Hence, the momentum map with respect to the natural integration pairing is given by
\begin{equation}
	\SectionSpaceAbb{J}(g, \pi) = \left( 2 (\divergence_g \pi)^\flat, 0 \right) \in \DiffFormSpace^1(\Sigma, \VolSpace) \times \VolSpace(\Sigma).
\end{equation}
The fiber derivative \( \difpFrac{\SectionSpaceAbb{L}}{h}: \TBundle \MetricSpace(\Sigma) \times \mathrm{L}\SectionSpaceAbb{G} \to \CotBundle \MetricSpace(\Sigma) \) is given by
\begin{equation}
	\label{eq:generalRelativity:legendreTransformation}
	\difpFrac{\SectionSpaceAbb{L}}{h}
		= \frac{1}{2 \ell} \bigl( \,g(h, \cdot) - \tr_g h \, \tr_g (\cdot) \bigr) \vol_g.
\end{equation}
Moreover, we have
\begin{equation}
	\difpFrac{\SectionSpaceAbb{L}}{S} = 0,
	\qquad
	\difpFrac{\SectionSpaceAbb{L}}{\ell} = \left( \RicScalar_{g} - \frac{1}{4 \ell^2} \norm{h}^2_{g} + \frac{1}{4 \ell^2} (\tr_{g} h)^2\right) \, \vol_{g}.
\end{equation}
Using these identities, a straightforward calculation shows that the momentum map constraint~\eqref{eq:clebschLagrange:hamiltonian:constraint} takes the form
\begin{subequations}\label{eq:generalRelativity:constraints}\begin{align}
	\label{eq:generalRelativity:constraints:diffeo}
	\divergence_g k - \grad_g (\tr_g k) = 0,
	\\
	\label{eq:generalRelativity:constraints:hamiltonian}
	0 = \RicScalar_{g} - \norm{k}^2_{g} + (\tr_{g} k)^2.
\end{align}\end{subequations}
Equations~\eqref{eq:generalRelativity:constraints:diffeo} and~\eqref{eq:generalRelativity:constraints:hamiltonian} are called the \emphDef{diffeomorphism constraint} and the \emphDef{Hamiltonian constraint}, respectively.
Hence, in our framework, these constraints are derived from the momentum map constraint~\eqref{eq:clebschLagrange:hamiltonian:constraint}.

\begin{remark}
	We emphasize that the Hamiltonian constraint stands on a different footing than the diffeomorphism constraint.
	This is due to the fact that the \( \sFunctionSpace(\Sigma) \)-part of \( \SectionSpaceAbb{G} \) does not act on the configuration space.
	Whereas the diffeomorphism constraint arises from setting the \( \DiffGroup(\Sigma) \)-component of \( \SectionSpaceAbb{J} \) to zero, the Hamiltonian constraint results from \( \difpFrac{\SectionSpaceAbb{L}}{\ell} = 0 \).
	As a consequence, the diffeomorphism constraint is a real momentum map constraint.
		
	Related to these facts is the observation that the Hamiltonian constraint fails to generate a Lie algebra \parencite{BergmannKomar1972} although the lapse function is of course an element of the Lie algebra of vector fields on \( M \).
	Recently, \textcite{BlohmannFernandesWeinstein2013} argued that the correct algebraic setting incorporating the full symmetry of Einstein's equation is that of groupoids and algebroids.
	It would hence be interesting to extend the Clebsch--Lagrange variational principle from Lie group actions to actions of Lie groupoids.
\end{remark}

\subsection{Hamiltonian picture}
Finally, let us briefly discuss the Hamiltonian picture.
According to~\eqref{eq:hamiltonian} the Clebsch--Hamiltonian \( \SectionSpaceAbb{H}: \CotBundle \MetricSpace(\Sigma) \times \mathrm{L} \SectionSpaceAbb{G} \to \R \) is given as the Legendre transform
\begin{equation}
	\SectionSpaceAbb{H}(g, \pi, S, \ell) = \dualPair{\pi}{h} - \SectionSpaceAbb{L}(g, h, S, \ell),
\end{equation}
where \( h \in \SymTensorFieldSpace^2(\Sigma) \) is defined as a function of \( (g, \pi, S, \ell) \) via the relation \( \pi = \difpFrac{\SectionSpaceAbb{L}}{h} (g, h, S, \ell) \).
By~\eqref{eq:generalRelativity:legendreTransformation}, we have
\begin{equation}
	h^\flat = 2 \ell \bar{\pi} - \ell \tr_g \bar{\pi} \, \tr_g (\cdot),
\end{equation}
where \( \bar{\pi} \cdot \vol_g = \pi \).
Hence, we obtain
\begin{equation}
	\SectionSpaceAbb{H}(g, \pi, S, \ell)
		= \int_\Sigma \ell \left( \norm{\bar{\pi}}^2_{g} - \RicScalar_{g} - \frac{1}{2} (\tr_{g} \bar{\pi})^2\right) \, \vol_{g}.
\end{equation}
A straightforward calculation shows that the Clebsch--Hamilton equations~\eqref{eq:clebschLagrange:hamiltonian} relative to this Hamiltonian yield the Einstein equation in its dynamical \( (1+3) \)-form.
More precisely, the dynamical equations~\eqref{eq:clebschLagrange:hamiltonian:hamiltonian} coincide with standard dynamical equations of general relativity, see \eg \parencite[Equation~E.2.35 and~E.2.36]{Wald1984}, and the constraint equations~\eqref{eq:clebschLagrange:hamiltonian:constraint} yield the momentum and Hamiltonian constraints as in the standard ADM formalism, see \eg \parencite[Equation~E.2.34 and~E.2.33]{Wald1984}.

To establish the connection with the ADM formalism, we follow the discussion in \cref{Ham-Pic-gen} and introduce the extended configuration space by 
\begin{equation}
	\SectionSpaceAbb{Q}_\ext \defeq \MetricSpace(\Sigma) \times \VectorFieldSpace(\Sigma) \times \sFunctionSpace(\Sigma),
\end{equation}
whose elements are denoted by \( (g, S, l) \).
The moment map constraints~\eqref{eq:Constr} cut out the following subset of the extended phase space \( \CotBundle \SectionSpaceAbb{Q}_\ext \):
\begin{equation}
	\SectionSpaceAbb{M}_1 \equiv \CotBundle \MetricSpace(\Sigma) \times \VectorFieldSpace(\Sigma) \times \sFunctionSpace(\Sigma) \times \set{0} \times \set{0}.
\end{equation}
We will denote elements of \( \SectionSpaceAbb{M}_1 \) by \( (g, \pi, S, l) \).
By~\eqref{eq:compositeHamiltonian}, the extended Hamiltonian on \( \SectionSpaceAbb{M}_1 \) is given by
\begin{equation}\begin{split}
	\SectionSpaceAbb{H}_\ext(g, \pi, S, l) 
		&= \SectionSpaceAbb{H}(g, \pi, S, l) - \kappa\bigl(\SectionSpaceAbb{J}(g, \pi), (S, l)\bigr)\,
		\\
		&= \int_\Sigma \ell \left( \norm{\bar{\pi}}^2_{g} - \RicScalar_{g} - \frac{1}{2} (\tr_{g} \bar{\pi})^2\right) \vol_g - 2 g(\divergence_g \bar{\pi}, S) \vol_g.
\end{split}\end{equation}
In the presence of a boundary this Hamiltonian will be modified by boundary terms, see \parencite[Equation~20.173]{Blau2018}.
Note that \( \SectionSpaceAbb{H}_\ext \) is the usual Hamiltonian in the ADM formalism (\eg, \parencites{ArnowittDeserMisner1959}[Equation~241]{Giulini2015}).
According to \cref{Ham-Pic-gen}, Hamilton's equations for \( \SectionSpaceAbb{H}_\ext \) coincide with the Clebsch--Hamilton equations for \( \SectionSpaceAbb{H} \) and thus with Einstein's equation (in the \( (1+3) \)-splitting).
Finally, we note that both Hamiltonians \( \SectionSpaceAbb{H} \) and \( \SectionSpaceAbb{H}_\ext \) vanish on the subset cut out by the constraints~\eqref{eq:generalRelativity:constraints}.
This is in accordance with the fact that there is no absolute time in general relativity.
A non-vanishing Hamiltonian would generate time evolution with respect to an external time parameter leading to a violation of the general covariance of the theory.
Instead, dynamics is completely governed by the constraints and evolution appears as a \enquote{gauge flow}. 

\begin{remark}
	Finally, let us comment on the reduction by stages procedure, see \cref{sec:clebschLagrange:reductionStages}.
	We have seen that the physical action~\eqref{eq:generalRelativity:einsteinHilbert3D} exhibits the symmetry group \( \SectionSpaceAbb{G} = \DiffGroup(\Sigma) \lSemiProduct \sFunctionSpace(\Sigma) \).
	According to the discussion in \cref{sec:clebschLagrange:reductionStages}, one would thus like to reduce the symmetry of the system using a symplectic reduction with respect to \( \SectionSpaceAbb{G} \).
	We should note, however, that in \cref{sec:clebschLagrange:reductionStages} we made the assumption that the Lagrangian does not explicitly depend on the \( \xi \)-variables.
	This assumption does not hold for general relativity as \( \SectionMapAbb{L} \) depends on \( \ell \).
	Nonetheless, as a first step, one can restrict attention to the \( \DiffGroup(\Sigma) \)-subgroup and pass to the symplectic quotient \( \CotBundle \SectionSpaceAbb{Q} \sslash \DiffGroup(\Sigma) \).
Up to a delicate discussion of the singular strata of the $\DiffGroup(\Sigma)$-action, the latter quotient coincides with the cotangent bundle of the superspace of Wheeler.
	This reduction procedure corresponds to implementing the momentum map constraint.
	It appears that the Hamiltonian constraint does not admit a similar interpretation in terms of a symplectic reduction.
\end{remark}

\begin{remark}
	In contrast to our approach, \textcite{FischerMarsden1972} have taken the extended configuration space \( \MetricSpace(\Sigma) \times \SectionSpaceAbb{G} \) as the starting point of a geometric analysis of the Einstein equation as a Lagrangian system.
	Note that passing from \( \mathrm{L} \SectionSpaceAbb{G} \) to \( \SectionSpaceAbb{G} \) completely changes the role of the shift vector field and the lapse function: in our approach they are configuration variables whereas in \parencite{FischerMarsden1972} they are generalized velocities.
\end{remark}

%%%%%%%%%%%%%%%%%%%%%%%%%%%%%%%%%%%%%%%%%%%%%%%%%%%%%%%%%%%%%%%%%%%%%%%%%%%%%%%%%%%%%%%%%%%%%%%%%%%%%%%%%
\appendix
\section{Appendix: \texorpdfstring{$(1+3)$}{(1+3)}-decomposition of the Yang--Mills--Higgs system}
\label{sec:yangMillsHiggs:decomposition}

In this section, we determine the \( (1+3) \)-decomposition of the Yang--Mills--Higgs system and of the objects involved using a coordinate-independent, geometric language.

\paragraph[Slicing]{Slicing: }
Let us assume that the spacetime \( (M, \eta) \) is globally hyperbolic and time-oriented.
Then, there exists a $3$-dimensional manifold $\Sigma$ and  a diffeomorphism \( \iota: \R \times \Sigma \to M \) such that
\begin{equation}
	\label{eq:splitting:metric}
	\iota^* \eta = - \ell(t)^2 \dif t^2 + g(t),
\end{equation}
where \( \ell \) is a smooth time-dependent positive function on \( \Sigma \) and \( g \) is a time-dependent Riemannian metric on \( \Sigma \), see \parencite[Theorem~1.3.10]{BarGinouxEtAl2007}.
For every \( t \in \R \), we denote by \( \iota_t: \Sigma \to M \) the induced embedding and by \( \Sigma_t \) the image of \( \Sigma \) under \( \iota_t \).
The submanifold $\Sigma_t$ is a Cauchy hypersurface at $t$. 
In what follows, we assume $\Sigma$ to be compact and without boundary.

Following \parencite[Section~6A]{GotayIsenbergMarsden2004}, a \emphDef{slicing} of \( P \) over \( \iota \) is a principal \( G \)-bundle \( \vec{\pi}: \vec{P} \to \Sigma \) and a principal bundle isomorphism \( \hat{\iota} \) fitting into the following diagram:
\begin{equationcd}
	\R \times \vec{P} \to[r, "\hat{\iota}"] \to[d, "\id_\R \times \vec{\pi}", swap]
		& P \to[d, "\pi"]
	\\
	\R \times \Sigma \to[r, "\iota"]
		& M.
\end{equationcd}
By \parencite[Corollary~4.9.7]{Husemoller1966} such a slicing always exists for small times; which suffices for the study of the initial value problem.
In order to simplify the presentation, let us assume that the slicing exists for all \( t \in \R \).

\paragraph[Decomposition of the objects]{Decomposition of the objects: }
Since the Yang--Mills--Higgs equations are functorial, we may suppress the bundle isomorphism \( \hat{\iota} \) and restrict attention to the case \( P = \R \times \vec{P} \) and \( M = \R \times \Sigma \).
Let \( \difp_t \) be the canonical timelike vector field on \( M \) determined by the splitting.
With a slight abuse of notation, we will denote the time vector field on \( P \) by \( \difp_t \), too.
Using this notation, the tangent bundles to \( M \) and \( P \) split into timelike and spacelike directions as follows:
\begin{align}
	\TBundle_{t,x} M &\isomorph \R \, \difp_t \oplus \TBundle_x \Sigma, \\
	\TBundle_{t,p} P &\isomorph \R \, \difp_t \oplus \TBundle_p \vec{P}. \label{eq:splitting:P}
\end{align}
These decompositions yield decompositions of the geometric objects living on \( \Sigma \) and \( P \), respectively.
For example, a vector field on \( M \) decomposes into \( X = X^0 \, \difp_t + \vec{X} \), where \( X^0(t) \in \sFunctionSpace(\Sigma) \) and \( \vec{X}(t) \in \sSectionSpace(\TBundle \Sigma ) \) for every \( t \).
The complex of differential forms on \( M \) refines to a bigraded complex \( \DiffFormSpace^{\cdot, \cdot}(M) \).
Indeed, we can write an arbitrary \( k \)-form \( \alpha \) as
\begin{equation}\label{eq:splitting:diffForm}
	\alpha = \alpha_0 \wedge \dif t + \vec{\alpha},
\end{equation}
with \( \alpha_0(t) \in \DiffFormSpace^{k-1}(\Sigma) \) and \( \vec{\alpha}(t) \in \DiffFormSpace^k(\Sigma) \).
A similar decomposition holds for differential forms with values in an associated bundle \( F = P \times_G \FibreBundleModel{F} \).
In this case, \( \alpha_0 \) and \( \vec{\alpha} \) are time-dependent differential forms on \( \Sigma \) with values in the bundle \( \vec{F} = \vec{P} \times_G \FibreBundleModel{F} \).
In the sequel, the restriction of the differential \( \dif \) on \( M \) to \( \Sigma \) will be denoted by \( \dif_\Sigma \).
A similar notation is used for the submanifold \( \vec{P} \subset P \).
A straightforward calculation shows that the exterior differential of \( \alpha \in \DiffFormSpace^k(\Sigma) \) decomposes as
\begin{equation}\label{eq:splitting:decomposeExtDiff}
	\dif \alpha = \left(\dif_\Sigma \alpha_0 + (-1)^k \, \dot{\vec{\alpha}}\right) \wedge \dif t + \dif_\Sigma \vec{\alpha},
\end{equation}
where the time-dependent \( k \)-form \( \dot{\vec{\alpha}} \) on \( \Sigma \) is defined by 
\begin{equation}
	\dot{\vec{\alpha}} (X_1, \ldots, X_k) \defeq \difp_t (\vec{\alpha}(X_1, \ldots, X_k))
\end{equation}
for \( X_i \in \TBundle \Sigma \).
A connection \( A \) in \( P \) can be seen as a \( G \)-equivariant \( 1 \)-form on \( P \) with values in \( \LieA{g} \) and thus, by~\eqref{eq:splitting:P}, decomposes as
\begin{equation}
	A = A_0 \dif t + \vec{A}
\end{equation}
into \( A_0(t) \in \sSectionSpace(\AdBundle \vec{P}) \) and a time-dependent connection \( \vec{A}(t) \) in \( \vec{P} \).
Clearly, \( A_0 = A(\difp_t) \) has the geometric meaning of measuring how much \( \difp_t \) fails to be horizontal with respect to \( A \).
The horizontal lift \( X^A_{t,p} \) of a vector \( X = X^0(t) \difp_t + \vec{X}(t) \in \TBundle_{t,x} M \) to the point \( (t,p) \in P \) has the form
\begin{equation}
	X^A_{t,p} = X^0(t) \, \difp_t + \bigl(\vec{X}^{\vec{A}}_p(t) - X^0(t) \, A_0(t,p) \ldot p\bigr),
\end{equation}
where \( \vec{X}^{\vec{A}}(t) \) denotes the horizontal lift of \( \vec{X}(t) \) with respect to the induced connection \( \vec{A}(t) \) in \( \vec{P} \).
Indeed, \( X^A_{t,p} \) is clearly a lift of \( X_{t,x} \) and, since \( A(X^A) = X^0 A_0 + 0 - X^0 A_0 = 0 \), it is horizontal.

Let us compute the decomposition of the covariant exterior differential.
Recall that, for every \( \alpha \in \DiffFormSpace^k(M, F) \), there exists a unique horizontal \( k \)-form \( \tilde{\alpha} \in \DiffFormSpace^k(P, \FibreBundleModel{F})^G \) that projects onto \( \alpha \).
Similarly to the reasoning above, every \( G \)-equivariant form \( \tilde{\alpha} \in \DiffFormSpace^k(P, \FibreBundleModel{F})^G \) can be written as a sum,
\begin{equation}
	\tilde{\alpha} = \tilde{\alpha}_0 \wedge \dif t + \vec{\tilde{\alpha}}, 
\end{equation}
where \( \tilde{\alpha}_0(t) \in \DiffFormSpace^{k-1}(\vec{P}, \FibreBundleModel{F})^G \) and \( \vec{\tilde{\alpha}}(t) \in \DiffFormSpace^k(\vec{P}, \FibreBundleModel{F})^G \).
Decomposing \( \dif \tilde{\alpha} \) similarly to~\eqref{eq:splitting:decomposeExtDiff}, we calculate
\begin{equation}\begin{split}
	\dif_A \tilde{\alpha}
		&= \dif \tilde{\alpha} + A \wedgeldot \tilde{\alpha}
		\\
		&= \left(\dif_{\vec{P}} \, \tilde{\alpha}_0 + (-1)^k \, \dot{\vec{\tilde{\alpha}}}\right) \wedge \dif t + \dif_{\vec{P}} \, \vec{\tilde{\alpha}} + (A_0 \dif t + \vec{A}) \wedgeldot (\tilde{\alpha}_0 \wedge \dif t + \vec{\tilde{\alpha}})
		\\
		&= \left(\dif_{\vec{A}} \tilde{\alpha}_0 + (-1)^k \, \dot{\vec{\tilde{\alpha}}} + (-1)^{k} A_0 \ldot \vec{\tilde{\alpha}}\right) \wedge \dif t + \dif_{\vec{A}} \vec{\tilde{\alpha}}.
\end{split}\end{equation}
Introducing the covariant time-derivative \( \partial_t^{A_0} \beta \defeq \dot{\beta} + A_0 \ldot \beta \) for a time-dependent \( \beta(t) \in \DiffFormSpace^k(\Sigma, \vec{F}) \), we thus find
\begin{equation}\label{eq:splitting:covExtDiff}
	\dif_A \alpha =
		\left(\dif_{\vec{A}} \alpha_0 + (-1)^k \partial_t^{A_0}\vec{\alpha}\right) \wedge \dif t + \dif_{\vec{A}} \vec{\alpha}.
\end{equation}
In particular, in degree \( 0 \), a section \( \varphi \in \sSectionSpace(F) \) is the same as a time-dependent section of \( \vec{F} \), which we continue to denote by \( \varphi \).
We then have
\begin{equation}
	\label{eq:splitting:covExtDiff:section}
	\dif_A \varphi
		= \left(\partial_t^{A_0}\varphi\right) \dif t + \dif_{\vec{A}} \varphi.
\end{equation}

A similar strategy also works for the decomposition of the curvature form.
Indeed, the structure equation for the connection \( A = A_0 \dif t + \vec{A} \) combined with the decomposition~\eqref{eq:splitting:decomposeExtDiff} of the exterior differential yields
\begin{equation}\begin{split}
	F_A  
		&= \dif A + \frac{1}{2} \wedgeLie{A}{A}
		\\
		&= \left(\dif_{\vec{P}} \, A_0 - \dot{\vec{A}}\right) \wedge \dif t + \dif_{\vec{P}} \, \vec{A} 
			+ \frac{1}{2} \wedgeLie*{(A_0 \dif t + \vec{A})}{(A_0 \dif t + \vec{A})}
		\\
		&= \left(\dif_{\vec{P}} \, A_0 - \dot{\vec{A}} + \LieBracket{\vec{A}}{A_0}\right) \wedge \dif t + \dif_{\vec{P}} \, \vec{A} + \frac{1}{2} \wedgeLie{\vec{A}}{\vec{A}}
		\\
		&= \left(\dif_{\vec{A}} \, A_0 - \dot{\vec{A}}\right) \wedge \dif t + F_{\vec{A}}.
\end{split}\end{equation}
As usual, we introduce the time-dependent color-electric field \( E \defeq \dif_{\vec{A}} \, A_0 - \dot{\vec{A}} \) and the color-magnetic field \( B \defeq F_{\vec{A}} \).
Then, the decomposition of the curvature takes the form
\begin{equation}
	\label{eq:splitting:curvature}
 	F_A = E \wedge \dif t + B \,.
\end{equation}

We also need to decompose the Hodge-star operator $\hodgeStar$ corresponding to $\eta$.
From now on, we denote it by $\hodgeStar_\eta $ as opposed to the Hodge star corresponding to the time-dependent Riemannian metric $g$, which we denote by $\hodgeStar_g $.
In the usual convention (see \eg \parencite[Definition~4.5.1]{RudolphSchmidt2012}), the Hodge star operator contains a minus sign depending on the signature of the metric:
\begin{alignat}{2}
	\hodgeStar_\eta \alpha 
		&= - \alpha^{\sharp_\eta} \contr \vol_\eta,
		\qquad
		&&\alpha \in \DiffFormSpace^k (M)\, ,
	\\
	\hodgeStar_g \beta 
		&= \beta^{\sharp_g} \contr \vol_g,
		\qquad
		&&\beta \in \DiffFormSpace^l (\Sigma)\, ,
\end{alignat}
where \( \alpha^{\sharp_\eta} \) is the \( k \)-vector field associated to \( \alpha \) via the metric \( \eta \).
Furthermore, we assume that \( M \) be time-oriented in such a way that 
\begin{equation}
	\vol_\eta = \ell \dif t \wedge \vol_{g},
\end{equation}
where the lapse function \( l(t) \) was introduced in~\eqref{eq:splitting:metric}.
Let \( \beta(t) \in \DiffFormSpace^k(\Sigma) \) be a time-dependent \( k \)-form on \( \Sigma \).
Since the metric \( \eta \) has positive signature in the spacelike directions, we have \( \beta^{\sharp_\eta} = \beta^{\sharp_{g}} \).
Thus, we obtain
\begin{equation}\begin{split}
	\hodgeStar_\eta \beta
		&= - \beta^{\sharp_\eta} \contr \vol_\eta
		\\
		&= - \beta^{\sharp_g} \contr (\ell \dif t \wedge \vol_g)
		\\
		&= (-1)^{k+1} \ell \dif t \wedge (\beta^{\sharp_g} \contr \vol_g)
		\\
		&= (-1)^{k+1} \ell \dif t \wedge \hodgeStar_g \beta
		\\
		&= \ell (\hodgeStar_g \beta) \wedge \dif t .
\end{split}\end{equation}
Take \( \alpha = \alpha_0 \wedge \dif t + 
\vec{\alpha} \in \DiffFormSpace^k(M, F) \). 
Using the identity \( \hodgeStar_\eta \, (\beta \wedge \dif t) = (\dif t)^{\sharp_\eta} \contr \hodgeStar_\eta \beta \) and \( (\dif t)^{\sharp_\eta} = - \ell^{-2} \difp_t \), we find 
\begin{equation}\label{eq:splitting:hodgeStar}\begin{split}
	\hodgeStar_\eta \alpha
		&= - \ell^{-2}\difp_t \contr (\hodgeStar_{\eta} \alpha_0) + \hodgeStar_\eta \vec{\alpha}
		\\
		&= - \ell^{-2}\difp_t \contr (\ell (\hodgeStar_g \alpha_0) \wedge \dif t ) + \ell (\hodgeStar_g \vec{\alpha}) \wedge \dif t
		\\
		&= (-1)^{k+1} \ell^{-1} \hodgeStar_g \alpha_0 + \ell (\hodgeStar_g \vec{\alpha}) \wedge \dif t .
\end{split}\end{equation}

Finally, let us determine the decomposition of the diamond product defined in~\eqref{eq:yangMillsHiggs:defDiamond}.
Let \( \alpha \in \DiffFormSpace^k(M, F) \), \( \beta \in \DiffFormSpace^{4-r-k}(M, F^*) \) and \( \xi \in \DiffFormSpace^r(M, \AdBundle P) \).
Decomposing these objects with respect to~\eqref{eq:splitting:diffForm}, we obtain
\begin{equation}\begin{split}
	&\wedgeDual{\xi}{(\alpha \diamond \beta)}
		\\
		&= \wedgeDual{(\xi \wedgeldot \alpha)}{\beta}
		\\
		&= \wedgeDual{\left((\xi_0 \dif t + \vec{\xi}) \wedgeldot (\alpha_0 \wedge \dif t + \vec{\alpha})\right)}{\left(\beta_0 \wedge \dif t + \vec{\beta}\right)}
		\\
		&= \wedgeDual{(\xi_0 \dif t \wedgeldot \vec{\alpha})}{\vec{\beta}}
			+ (-1)^{k+r} \wedgeDual{(\vec{\xi} \wedgeldot \alpha_0)}{\vec{\beta}} \wedge \dif t
			+ \wedgeDual{(\vec{\xi} \wedgeldot \vec{\alpha})}{\beta_0} \wedge \dif t
		\\
		&= \wedgeDual{\xi_0 \dif t}{(\vec{\alpha} \diamond_\Sigma \vec{\beta})} 
			+ \wedgeDual{\vec{\xi}}{\left((-1)^{k+r} \, \alpha_0 \diamond_\Sigma \vec{\beta} + \vec{\alpha} \diamond_\Sigma \beta_0\right)} \wedge \dif t
		\\
		&= \wedgeDual{\xi}{\left(\vec{\alpha} \diamond_\Sigma \vec{\beta} + ((-1)^{k+r} \, \alpha_0 \diamond_\Sigma \vec{\beta} + \vec{\alpha} \diamond_\Sigma \beta_0) \wedge \dif t \right)},
\end{split}\end{equation}
where \( \diamond_\Sigma \) denotes the diamond operation on \( \Sigma \).
Hence, for \( \alpha \in \DiffFormSpace^k(M, F) \) and \( \beta \in \DiffFormSpace^{4-r-k}(M, F^*) \),
\begin{equation}\label{eq:splitting:diamondProd}
	\alpha \diamond \beta = \left((-1)^{k+r} \, \alpha_0 \diamond_\Sigma \vec{\beta} + \vec{\alpha} \diamond_\Sigma \beta_0\right) \wedge \dif t  + \vec{\alpha} \diamond_\Sigma \vec{\beta}.
\end{equation}

\paragraph[Decomposition of the equations of motion]{Decomposition of the equations of motion:}
We now turn our attention to the Yang--Mills--Higgs equations.
Using the decomposition of the Hodge operator~\eqref{eq:splitting:hodgeStar}, of the curvature~\eqref{eq:splitting:curvature} and of the exterior differential~\eqref{eq:splitting:covExtDiff:section}, the Lagrangian~\eqref{eq:yangMillsHiggs:action} can be written as:
\begin{equation}\label{eq:splitting:lagrangian}\begin{split}
	L(A, \varphi) 
		&= \frac{1}{2} \wedgeDual{F_A}{\hodgeStar_\eta F_A} + \frac{1}{2} \wedgeDual{\dif_A \varphi}{\hodgeStar_\eta \dif_A \varphi} - V(\varphi) \vol_\eta
		\\
		&= \frac{1}{2} \wedgeDual{\bigl(E \wedge \dif t + B\bigr)}{\bigl(-\ell^{-1}\hodgeStar_g E + \ell (\hodgeStar_g B) \wedge \dif t\bigr)} 
		\\
		&\quad + \frac{1}{2} \wedgeDual{\bigl((\partial_t^{A_0} \varphi) \, \dif t + \dif_{\vec{A}} \varphi\bigr)}{\bigl(\ell^{-1} \hodgeStar_g (\partial_t^{A_0} \varphi) + \ell (\hodgeStar_g \dif_{\vec{A}} \varphi) \wedge \dif t\bigr)} \\
		&\quad - \ell \, V(\varphi) \dif t \wedge \vol_g
		\\
		&= \dif t \wedge \bigg(\frac{1}{2 \ell} \wedgeDual{E}{\hodgeStar_{g} E} - \frac{\ell}{2}\wedgeDual{B}{\hodgeStar_g B}
		\\
		&\hphantom{{}=\dif t \wedge } \quad + \frac{1}{2 \ell} \wedgeDual{\partial_t^{A_0} \varphi}{\hodgeStar_g \partial_t^{A_0} \varphi} - \frac{\ell}{2} \wedgeDual{\dif_{\vec{A}} \varphi}{\hodgeStar_g \dif_{\vec{A}} \varphi} - \ell \, V(\varphi) \vol_g \bigg).
\end{split}\end{equation}

Now, the decomposition of the field equations may be accomplished either by taking the variation of this action or by directly inserting the above decompositions of the fields into the covariant field equations~\eqref{eq:yangMillsHiggs4d}.
The first approach is the subject of \cref{sec:yangMillls:asClebschLagrange} and, here, we take the second, more direct option.
\begin{prop}[Decomposition of the Yang--Mills--Higgs equations]
	\label{prop:splitting:yangMillsHiggs}
	The Yang--Mills--Higgs equations~\eqref{eq:yangMillsHiggs4d} decompose into
	\begin{subequations}\label{eq:splitting:yangMillsHiggs3dEvol}\begin{align+}
		\dif_{\vec{A}} (\ell \hodgeStar_g B) - \partial_t^{A_0} (\ell^{-1} \hodgeStar_{g} E) &= - \ell \, \varphi \diamond_\Sigma (\hodgeStar_g \dif_{\vec{A}} \varphi),
		\label{eq:splitting:yangMillsHiggs3d:evol:ampere}
		\\
		\dif_{\vec{A}} (\ell^{-1} \hodgeStar_{g} E) &= \ell^{-1} \, \varphi \diamond_\Sigma (\hodgeStar_g \partial_t^{A_0} \varphi),
		\label{eq:splitting:yangMillsHiggs3dConstraint}
		\\
		- \partial_t^{A_0} (\ell^{-1} \hodgeStar_g \partial_t^{A_0} \varphi) + \dif_{\vec{A}} (\ell \hodgeStar_g \dif_{\vec{A}} \varphi) &= \ell \, V' (\varphi)\vol_g .
		\label{eq:splitting:yangMillsHiggs3d:evol:higgs}
	\end{align+}\qedhere\end{subequations}
\end{prop}
\begin{proof}
Using the decomposition formulae for the covariant exterior derivative~\eqref{eq:splitting:covExtDiff} and for the Hodge operator~\eqref{eq:splitting:hodgeStar}, we calculate
\begin{subequations}\begin{align}
	\begin{split}
	\dif_A \hodgeStar_\eta F_A 
		&= \dif_A \hodgeStar_\eta (E \wedge \dif t + B)
		\\
		&= \dif_A \left(\ell (\hodgeStar_g B) \wedge \dif t - \ell^{-1} \hodgeStar_g E \right)
		\\
		&=  \left(\dif_{\vec{A}} (\ell \hodgeStar_g B) - \partial_t^{A_0} (\ell^{-1} \hodgeStar_{g} E)\right) \wedge \dif t - \dif_{\vec{A}} (\ell^{-1} \hodgeStar_{g} E),
	\end{split}
\intertext{Similarly, using~\eqref{eq:splitting:diamondProd}, we find}
	\begin{split}
	\varphi \diamond \hodgeStar_\eta \dif_A \varphi 
		&= \varphi \diamond \hodgeStar_\eta \left((\partial_t^{A_0} \varphi) \, \dif t + \dif_{\vec{A}} \varphi\right)
		\\
		&= \varphi \diamond \bigl(\ell^{-1} \hodgeStar_g \partial_t^{A_0} \varphi + \ell (\hodgeStar_g \dif_{\vec{A}} \varphi) \wedge \dif t\bigr)
		\\
		&= \ell^{-1} \, \varphi \diamond_\Sigma (\hodgeStar_g \partial_t^{A_0} \varphi) + \ell \bigl(\varphi \diamond_\Sigma (\hodgeStar_g \dif_{\vec{A}} \varphi)\bigr) \wedge \dif t .
	\end{split}
\end{align}\end{subequations}
Comparing these equations with~\eqref{eq:yangMillsHiggs:4d:ym}, we obtain~\eqref{eq:splitting:yangMillsHiggs3d:evol:ampere} and~\eqref{eq:splitting:yangMillsHiggs3dConstraint}.
Finally,
\begin{equation}
	\begin{split}
	\dif_A \hodgeStar_\eta \dif_A \varphi 
		&= \dif_A (\ell^{-1} \hodgeStar_g \partial_t^{A_0} \varphi + \ell \hodgeStar_g \dif_{\vec{A}} \varphi \wedge \dif t)
		\\
		&= \left(- \partial_t^{A_0} (\ell^{-1} \hodgeStar_g \partial_t^{A_0} \varphi) + \dif_{\vec{A}} (\ell \hodgeStar_g \dif_{\vec{A}} \varphi) \right) \wedge \dif t
	\end{split}
\end{equation}
shows that~\eqref{eq:yangMillsHiggs:4d:higgs} is equivalent to the evolution equation~\eqref{eq:splitting:yangMillsHiggs3d:evol:higgs} for the Higgs field.
\end{proof}
\begin{remark}
	Using~\eqref{eq:splitting:covExtDiff} and~\eqref{eq:splitting:curvature}, the Bianchi identity for \( A \) reads:
	\begin{equation}
		0 = \dif_A F_A 
			= \left(\dif_{\vec{A}} E + \partial_t^{A_0}B\right) \wedge \dif t + \dif_{\vec{A}} B.
	\end{equation}
	This implies Faraday's law and the Bianchi identity for \( \vec{A} \):
	\begin{align}
			\dif_{\vec{A}} E + \partial_t^{A_0} B &= 0,
			\label{eq:yangMillsHiggs3d:evol:faraday}
			\\
			\dif_{\vec{A}} B &= 0.
			\label{eq:yangMillsHiggs3d:evol:bianchi}
	\end{align}
	Note that this set of equations together with~\eqref{eq:splitting:yangMillsHiggs3d:evol:ampere} and~\eqref{eq:splitting:yangMillsHiggs3dConstraint} constitute the non-abelian counterpart of Maxwell's equations (coupled to matter).
\end{remark}

\begin{refcontext}[sorting=nyt]{}
	\printbibliography
\end{refcontext}

\end{document}